\documentclass[journal]{IEEEtran}
\usepackage{graphicx} 

\usepackage{amsmath}
\usepackage{amsfonts}
\usepackage{color}
\usepackage{soul}
\usepackage{cite}
\usepackage{graphicx}
\usepackage{algorithm}
\usepackage[noend]{algpseudocode}
\usepackage{algorithmicx}
\usepackage{mathtools}
\usepackage{bm}
\usepackage{amsthm}

\usepackage{nicematrix}
\usepackage{booktabs}
\usepackage{multirow}

\newtheorem{theorem}{Theorem}

\newtheorem{assumption}{Assumption}
\newtheorem{lemma}{Lemma}

\newtheorem{proposition}{Proposition}
\newtheorem{defn}{Definition}

\DeclareMathSymbol{\shortminus}{\mathbin}{AMSa}{"39}

\DeclareMathOperator*{\argmax}{argmax}
\DeclareMathOperator*{\argmin}{argmin}

\newcommand{\real}{\mathbb{R}}
\newcommand{\integer}{\mathbb{Z}}

\algnewcommand{\LineComment}[1]{\State \(\triangleright\) #1}
\algnewcommand\algorithmicswitch{\textbf{switch}}
\algnewcommand\algorithmiccase{\textbf{case}}
\algnewcommand\algorithmicassert{\texttt{assert}}
\algnewcommand\Assert[1]{\State \algorithmicassert(#1)}%

\algdef{SE}[SWITCH]{Switch}{EndSwitch}[1]{\algorithmicswitch\ #1\ \algorithmicdo}{\algorithmicend\ \algorithmicswitch}%
\algdef{SE}[CASE]{Case}{EndCase}[1]{\algorithmiccase\ #1}{\algorithmicend\ \algorithmiccase}%
\algdef{SE}[DOWHILE]{Do}{doWhile}{\algorithmicdo}[1]{\algorithmicwhile\ #1}%
\algtext*{EndSwitch}%
\algtext*{EndCase}%

\usepackage{caption} 

\title{Mixed-Integer MPC-Based Motion Planning Using Hybrid Zonotopes with Tight Relaxations}
\author{Joshua A. Robbins, Jacob A. Siefert, Sean Brennan, and Herschel C. Pangborn
\thanks{Joshua A. Robbins, Sean Brennan, and Herschel C. Pangborn are with the Department of Mechanical Engineering, The Pennsylvania State University, University Park, PA 16802 USA (e-mail: {\tt\small jrobbins@psu.edu, sbrennan@psu.edu, hcpangborn@psu.edu}).}
\thanks{Jacob A. Siefert is with the Penn State Applied Research Laboratory, University Park, PA 16802 USA (e-mail: {\tt\small jas7031@psu.edu}).}
\thanks{This research was supported by Peraton.}
}
\date{September 2024}

\begin{document}

\maketitle

\begin{abstract} 
Autonomous vehicle (AV) motion planning problems often involve non-convex constraints, which present a major barrier to applying model predictive control (MPC) in real time on embedded hardware. This paper presents an approach for efficiently solving mixed-integer MPC motion planning problems using a hybrid zonotope representation of the obstacle-free space. The MPC optimization problem is formulated as a multi-stage mixed-integer quadratic program (MIQP) using a hybrid zonotope representation of the non-convex constraints. Risk-aware planning is supported by assigning costs to different regions of the obstacle-free space within the MPC cost function. A multi-stage MIQP solver is presented that exploits the structure of the hybrid zonotope constraints. For some hybrid zonotope representations, it is shown that the convex relaxation is \emph{tight}, i.e., equal to the convex hull. In conjunction with logical constraints derived from the AV motion planning context, this property is leveraged to generate tight quadratic program (QP) sub-problems within a branch-and-bound mixed-integer solver. The hybrid zonotope structure is further leveraged to reduce the number of matrix factorizations that need to be computed within the QP sub-problems. Simulation studies are presented for obstacle-avoidance and risk-aware motion planning problems using polytopic maps and occupancy grids. In most cases, the proposed solver finds the optimal solution an order of magnitude faster than a state-of-the-art commercial solver. Processor-in-the-loop studies demonstrate the utility of the solver for real-time implementations on embedded hardware.
\end{abstract}


\section{Introduction}
Motion planning is a foundational task in vehicle autonomy. Often, motion planning algorithms serve as an intermediate autonomy level between high-level decision making and/or path-planning algorithms and low-level path-following controllers \cite{matni2024quantitative}. Given some local description of the environment, which is in general non-convex \cite{ioan2021mixed}, these algorithms must construct collision-free trajectories and respect constraints on the vehicle's motion \cite{gautam2024overview}. Additional considerations may factor into the motion plan such as smoothness or comfort \cite{gautam2024overview}, and uncertainty mitigation or risk reduction \cite{benciolini2024combining, safaoui2021risk}. Motion planning is generally performed online, which makes computational efficiency a necessity.

\subsection{Gaps in the Literature}
MPC is a widely used technique in AV motion planning. Trajectories generated by MPC-based motion planners have numerous features including dynamic feasibility, constraint satisfaction, and optimality in a receding horizon sense. A significant barrier to the widespread adoption of MPC-based AV motion planners is the computational challenge of solving MPC problems of sufficient expressiveness in real-time on embedded hardware \cite{gautam2024overview}.

Many real-time implementable approaches for MPC-based motion planning rely on approximate or local solutions. Examples include sequential convex programming \cite{augugliaro2012generation, alrifaee2017sequential, morgan2014model}, model predictive path integral control (MPPI) \cite{williams2016aggressive, williams2017model}, and model predictive contouring control \cite{brito2019model}. A limitation of these approaches is that they are not guaranteed to find the globally optimal solution when there are non-convex constraints, as in the case of obstacle avoidance \cite{ioan2021mixed}. 

Several MPC formulations have been proposed to account for non-convexity in motion planning problems. For example, branching and scenario MPC account for non-convexity via enumerated scenarios and scenario trees \cite{oliveira2023interaction, sopasakis2019risk, chen2022interactive}. In \cite{nair2022collision}, collision avoidance is enforced in MPC via a dual formulation. Ref.~\cite{gratzer2024two} presents a mixed-integer MPC formulation for AV motion planning which is subsequently validated using high fidelity traffic simulations. A review of mixed-integer programming formulations for motion planning is given in \cite{ioan2021mixed}.  

MPC problems formulated as mixed-integer convex programs (e.g., \cite{gratzer2024two, quirynen2023real}) can be solved to global optimality using branch-and-bound methods \cite{floudas1995nonlinear}. Mixed-integer programs are often solved using general-purpose mixed-integer solvers, such as Gurobi \cite{gurobi} and MOSEK \cite{mosek}. Specialized techniques for solving  mixed-integer programs have also been proposed. In \cite{quirynen2023real}, a custom MIQP solver for motion planning is proposed and evaluated in robotics experiments. When compared to commercial mixed-integer solvers, the solver proposed in that study was found to be approximately 5-6 times faster than MOSEK and 1.5-2.5 times slower than Gurobi. Follow-on work by the same authors used a neural network to predict the optimal integer variables for the MIQPs, thus reducing the MIQPs to convex QPs at the expense of sub-optimality \cite{reiter2024equivariant}. Similar neural network-based approaches were proposed in \cite{cauligi2020learning}.

Advanced set representations have been leveraged for use in motion planning problem formulations. Zonotopes are used to reduce the complexity of a hyperplane arrangement description of the obstacle avoidance constraints in \cite{ioan2019complexity}. The obstacle-free space is described exactly using polynomial zonotopes in a nonlinear MPC formulation in \cite{nascimento2023nmpc}, and using hybrid zonotopes in a mixed-integer formulation in \cite{bird2021unions}. In both \cite{nascimento2023nmpc} and \cite{bird2021unions}, the advanced set representations resulted in reductions in optimization times when compared to a hyperplane arrangement description of the obstacle-free space. Both of these papers used general-purpose solvers rather than specialized motion planning solvers as in \cite{quirynen2023real, reiter2024equivariant}.

Many of the aforementioned problem formulations (e.g., \cite{quirynen2023real, nair2022collision,  gratzer2024two}) are based on a description of the obstacles and do not explicitly use a free space description. Formulating motion planning problems in terms of the obstacle-free space enables many useful problem specifications \cite{marcucci2024shortest, ioan2021mixed}, such as the ability to formulate the motion planning problem using an occupancy grid map (OGM). OGMs are widely used in robotics to fuse diverse sources of uncertainty within a unified, probabilistic description of the environment \cite{elfes1989using, macenski2023desks, hoermann2018dynamic}. Occupancy probabilities can be generalized to risks or costs as described in \cite{macenski2023desks}.

Despite their ubiquity, most existing approaches for motion planning over OGMs cast the problem in terms of a binary formulation where occupancy probabilities are not directly utilized. For instance, chance constraints are used to construct a binary OGM in a stochastic MPC formulation in \cite{brudigam2021stochastic}. Clothoid tentacles are used for motion planning over binary OGMs in \cite{mouhagir2016integrating}. In \cite{cho2023model}, a barrier function is used with nonlinear MPC to prohibit an AV from entering cells above a certain occupancy probability. In \cite{zhang2023real}, time-varying OGMs are used to calculate the feasible space for a sampling-based motion planner. For situations with significant uncertainty or for which there is no feasible way for the AV to avoid entering a cell with elevated occupancy probability, these approaches can be limiting. Ref.~\cite{artunedo2020motion} directly incorporates occupancy probabilities into a sampling-based motion planner that uses trajectory parameterizations. To the best of the authors' knowledge, there are no prior publications showing MPC-based motion planning where cell occupancy probabilities or costs are directly used in the MPC cost function.

\subsection{Contributions}
This article presents an approach to formulating and solving MIQPs for MPC-based motion planning. A hybrid zonotope set representation is used to represent the obstacle-free space, and its structure is exploited within the MIQP solver. The obstacle-free space is efficiently represented either as a general polytopic map or an OGM. Occupancy probabilities or costs associated with obstacle-free regions are incorporated into the MPC formulation.

This article builds upon our previous work \cite{robbins2024efficient}, which presented a branch-and-bound algorithm based on a notion of reachability between regions of the obstacle-free space and showed how the hybrid zonotope structure can be exploited within the QP sub-problems. The new contributions in this paper can be categorized as follows; 
(1)~\textit{Convex relaxations}:
We show that certain hybrid zonotopes constructed from general polytopes and from OGMs have the property that their convex relaxation is their convex hull. We use this property in combination with ``reachability'' constraints to tighten the QP sub-problems such that convergence can be reached with fewer branch-and-bound iterations. 
(2)~\textit{Algorithm development}: The branch-and-bound solver is modified to support multi-threading and warm-starting, and nodes are generated in such a way as to reduce the sensitivity of solution times to map complexity. The MPC formulation and MIQP solver are implemented in C++.
(3)~\textit{Region-dependent costs}: In order to perform risk-aware planning, region-dependent costs are added to the MPC formulation. The obstacle avoidance-based branch-and-bound logic is modified accordingly. 
(4)~\textit{Numerical results}: The proposed approach is evaluated in desktop computer simulations and real-time processor-in-the-loop testing.
When compared to a conventional mixed-integer constraint representation (unions of halfspace representation polytopes using the Big-M method \cite{ioan2021mixed}) and the state-of-the-art mixed-integer solver Gurobi \cite{gurobi}, the proposed approach often finds the optimal solution one to two orders of magnitude faster.
\section{Preliminaries}  \label{sec:preliminaries}

\subsection{Notation} Unless otherwise stated, scalars are denoted by lowercase letters, vectors by boldface lowercase letters, matrices by uppercase letters, and sets by calligraphic letters. Vectors consisting entirely of zeroes and ones are denoted by $\mathbf{0} = \begin{bmatrix} 0 & \cdots & 0 \end{bmatrix}^T$ and $\mathbf{1} = \begin{bmatrix} 1 & \cdots & 1 \end{bmatrix}^T$, respectively. The identity matrix is denoted as $I$. Empty brackets $[]$ indicate that a matrix has row or column dimension of zero. Diagonal and block diagonal matrices are denoted as $\text{diag}([\cdot])$ and $\text{blkdiag}([\cdot])$, respectively. Vertex representation (V-rep) polytopes are denoted by their vertices as $\mathcal{P} = \{ \mathbf{v}_1, ..., \mathbf{v}_n\}$. Square brackets following a vector denote an indexing operation such that $\mathbf{x}[i] = x_i$. A variable $\mathbf{x}$ is box constrained if $\mathbf{x} \in [\mathbf{x}^l, \mathbf{x}^u]$. A matrix $G$ is said to be in ``box constraint form'' if there exist permutation matrices $P_r$ and $P_c$ such that $G = P_r \Bar{G} P_c$ where $\Bar{G} = \begin{bmatrix}
    -I^T & I^T
\end{bmatrix}^T$.
The function $\mathbf{f}_r(\mathbf{v})$ returns the non-zero one-based indices of $\mathbf{v}$. For instance, if $\mathbf{v} = \begin{bmatrix} 0 & 0.3 & 0.7 \end{bmatrix}$, then $\mathbf{f}_r(\mathbf{v}) = \begin{bmatrix} 2 & 3 \end{bmatrix}$.

\subsection{Motion Planning Problem Statement}
We consider the problem of planning collision-free trajectories for an AV over a receding horizon. A discrete-time formulation is used with time step $\Delta t$ and horizon length $N$. The time step index is denoted by $k$, with $k=0$ corresponding to the current time. The AV dynamics are modeled as the linear time-invariant (LTI) system $\mathbf{x}_{k+1} = A \mathbf{x}_k + B \mathbf{u}_k$ subject to convex state and input constraints $\mathbf{x}_k \in \mathcal{X} \subset \real^{n_x}$ and $\mathbf{u}_k \in \mathcal{U} \subset \real^{n_u}$, respectively. A convex terminal state constraint set is also imposed, such that $\mathbf{x}_N \in \mathcal{X}_N \subset \real^{n_x}$. The position states of the system, $\mathbf{y}_k$, are additionally subject to non-convex obstacle avoidance constraints. These are described in terms of the obstacle-free space $\mathcal{F}$ such that $\mathbf{y}_k \in \mathcal{F}$, where $\mathcal{F}$ is not a function of $k$. Regions of the obstacle-free space may be assigned a cost to disincentivize the planning of trajectories that pass through those regions. These costs may, for example, be informed by uncertainty distributions for the AV or obstacle positions.

This problem statement is exemplified in Fig.~\ref{fig:prob-statement}. Here, an AV is tasked with achieving a reference state $\mathbf{x}^r_k$ (the pink star) while navigating through an environment with four obstacles (the boulders). The obstacle-free space $\mathcal{F}$ is partitioned into a grid with region-dependent costs. Darker red cells have higher costs while more transparent cells have lower costs. 
The reference is assumed to be provided by a high-level decision-making algorithm as is often the case in layered control architectures \cite{matni2024quantitative}. In this example and in the numerical results (Sec.~\ref{sec:results}), the reference is constant over the planning horizon such that $\mathbf{x}_k^r = \mathbf{x}^r \; \forall k \in \{0, ..., N\}$. The planned trajectory for this example is displayed using blue dots.

\begin{figure}[tb]
    \centering
    \includegraphics[width=0.45\linewidth,angle =-90]{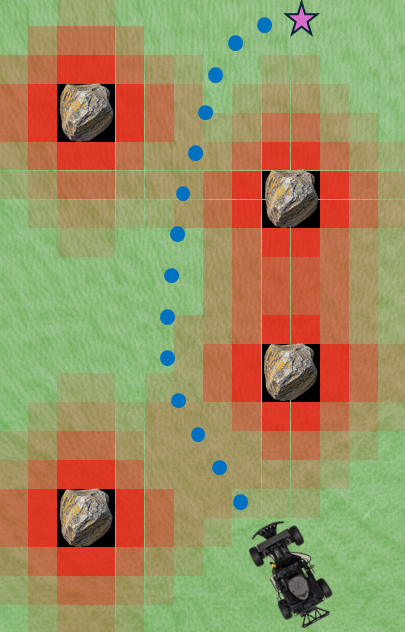}
    \caption{Problem overview for risk-aware motion planning.}
    \label{fig:prob-statement}
\end{figure}

\subsection{MPC Formulation}
Consider the MPC formulation
\begin{subequations} \label{eq:mpc-gen}
\begin{align}
    &\min_{\mathbf{x}_k, \mathbf{u_k}} \sum_{k=0}^{N-1} \left[ (\mathbf{x}_k - \mathbf{x}_k^r)^T Q_k (\mathbf{x}_k - \mathbf{x}_k^r) + \mathbf{u}_k^T R_k \mathbf{u}_k + \right. \nonumber \\
    & \qquad \left. q^r(\mathbf{y}_k) \right] + (\mathbf{x}_N - \mathbf{x}_N^r)^T Q_N (\mathbf{x}_N - \mathbf{x}_N^r) + q^r(\mathbf{y}_N) \;, \label{eq:mpc-gen-cost} \\
    &\text{s.t.} \; \forall k \in \mathcal{K} = \{0, \cdots, N-1 \}\;: \nonumber \\ 
    &\phantom{\text{s.t.}} \; \mathbf{x}_{k+1} = A \mathbf{x}_k + B \mathbf{u}_k\;, \\
    &\phantom{\text{s.t.}} \; \mathbf{y}_{k} = H \mathbf{x}_k,\; \mathbf{y}_{N} = H \mathbf{x}_N,\; \mathbf{x}[0] = \mathbf{x}_0 \;, \\
    &\phantom{\text{s.t.}} \; \mathbf{x}_k \in \mathcal{X},\; \mathbf{x}_N \in \mathcal{X}_N,\;\mathbf{u}_k \in \mathcal{U}\;, \label{eq:mpc-gen-state-input-cons} \\
    &\phantom{\text{s.t.}} \; \mathbf{y}_{k}, \mathbf{y}_{N} \in \mathcal{F} = \bigcup_{i=1}^{n_F} \mathcal{F}_i \subset \real^n \;, \label{eq:mpc-gen-obs-avoid} \\
    &\phantom{\text{s.t.}} \; \mathbf{y}_k \in \mathcal{F}_i \Rightarrow q^r(\mathbf{y}_k) = q_{i}^r
    \;. \label{eq:mpc-gen-reg-dep-costs}
\end{align}
\end{subequations}
The initial state of the system is given by $\mathbf{x}[0]$. The $H$ matrix extracts the AV position from $\mathbf{x}_k$. 
The function $q^r: \real^n \rightarrow \real$ couples the system outputs to region-dependent costs via \eqref{eq:mpc-gen-reg-dep-costs}. In this paper, region-dependent costs are derived from occupancy probabilities in an OGM. Denoting $p_i$ as the occupancy probability of cell $i$, the simple linear scaling $q_{i}^r = \kappa p_i$ is used, though more complex relationships between the occupancy probability and the MPC cost function are possible. 

Eq.~\eqref{eq:mpc-gen} is solved over a receding horizon, such that the current state $\mathbf{x}_0$ is provided to the MPC controller at each update time and the optimal input for the first step $\mathbf{u}_0$ is applied to the system. Alternatively, the input at the current time step may be constrained such that $\mathbf{u}[0] = \mathbf{u}_0$ is added as a constraint to \eqref{eq:mpc-gen}. In this case, the optimal input at time step $k=1$, $\mathbf{u}_1$, is stored and then applied to the system at the start of the next iteration. This formulation is advantageous for many practical real-time MPC implementations for which the time to solve the MPC optimization problem is not much less than the loop time of the controller \cite{aksland2022hierarchical}.  

The following assumptions are used throughout the rest of this paper. 

\begin{assumption} \label{as:Pk-diag}
$Q_k$, $R_k$, and $Q_N$ are diagonal with non-negative elements $\forall k \in \mathcal{K}$.
\end{assumption}
\begin{assumption} \label{as:X-U-XT}
$\mathcal{X}$, $\mathcal{U}$, and $\mathcal{X}_N$ are polytopes consistent with \cite[Def.~1]{scott2016constrained}.
\end{assumption}
\begin{assumption} \label{as:F-polytope}
$\mathcal{F}$ is a union of polytopes $\mathcal{F}_i$, $i \in \{1, ..., n_F\}$.
\end{assumption}
Assumption~\ref{as:Pk-diag} is required to exploit the structure of the hybrid zonotope constraints described in Sec.~\ref{sec:qp-cons-zono}.  Diagonal cost matrices are common in practice and can often be achieved under a change of variables if necessary \cite{lay2003linear}. Non-negative elements ensure that the MPC optimization problem is mixed-integer convex.
Assumption~\ref{as:X-U-XT} is typical for MPC formulations and Assumption~\ref{as:F-polytope} holds for obstacle avoidance problems with position bounds and polytopic obstacles.

\subsection{Hybrid Zonotope Representations of Obstacle-Free Space} \label{sec:preliminaries_zono}

We next briefly review the \emph{zonotope}, \emph{constrained zonotope}, and \emph{hybrid zonotope} set representations. A set $\mathcal{Z} \subset \real^n$ is a zonotope if $\exists \; G_c \in \real^{n \times n_g}$, $\mathbf{c} \in \real^{n}$ such that
\begin{equation} \label{eq:zonotope}
\mathcal{Z} = \left\{ G_c \bm{\xi}_c + \mathbf{c} \middle| \bm{\xi}_c \in \mathcal{B}_{\infty}^{n_g} \right\} \;,
\end{equation}
where $\mathcal{B}_\infty^{n_g} = \{\xi_c\in\mathbb{R}^{n_g} \mid \lVert \xi_c \rVert_\infty\leq1\}$ is the infinity-norm ball. Zonotopes are convex, centrally symmetric sets \cite{ziegler2012lectures}.

A set $\mathcal{Z_C} \subset \real^n$ is a constrained zonotope if $\exists \; G_c \in \real^{n \times n_g}$, $\mathbf{c} \in \real^{n}$, $A_c \in \real^{n_c \times n_g}$, $\mathbf{b} \in \real^{n_c}$ such that
\begin{equation} \label{eq:cons_zonotope}
\mathcal{Z_C} = \left\{ G_c \bm{\xi}_c + \mathbf{c} \middle| \bm{\xi}_c \in \mathcal{B}_{\infty}^{n_g},\; A_c \bm{\xi}_c = \mathbf{b} \right\} \;.
\end{equation} 
Constrained zonotopes can represent any polytope and can be constructed from a halfspace representation \cite[Thm 1]{scott2016constrained}. The shorthand notation $\mathcal{Z_C} = \langle G_c, \mathbf{c}, A_c, \mathbf{b} \rangle$ is used to denote a constrained zonotope.

Hybrid zonotopes extend \eqref{eq:cons_zonotope} by including binary factors $\bm{\xi}_b$. A set $\mathcal{Z_H} \subset \real^n$ is a hybrid zonotope if in addition to $G_c$, $\mathbf{c}$, $A_c$, and $\mathbf{b}$, $\exists\; G_b \in \real^{n \times n_b}$, $A_b \in \real^{n_c \times n_b}$ such that
\begin{equation} \label{eq:hyb_zonotope}
\mathcal{Z_H} = \left\{ \begin{bmatrix} G_c & G_b \end{bmatrix} 
\begin{bmatrix} \bm{\xi}_c \\ \bm{\xi}_b \end{bmatrix} + \mathbf{c} \middle| 
\begin{matrix}
\begin{bmatrix} \bm{\xi}_c \\ \bm{\xi}_b \end{bmatrix} \in \mathcal{B}_{\infty}^{n_g} \times \{-1,1\}^{n_b} \\
\begin{bmatrix} A_c & A_b \end{bmatrix} \begin{bmatrix} \bm{\xi}_c \\ \bm{\xi}_b \end{bmatrix} = \mathbf{b}
\end{matrix}
\right\} \;.
\end{equation}
Hybrid zonotopes can represent any union of polytopes \cite{bird2021unions, siefert2023reachability}. As such, by Assumption~\ref{as:F-polytope}, the obstacle-free space $\mathcal{F}$ in \eqref{eq:mpc-gen} can be represented as a hybrid zonotope. The shorthand notation $\mathcal{Z_H} = \langle G_c, G_b, \mathbf{c}, A_c, A_b, \mathbf{b} \rangle$ is used to denote a hybrid zonotope.

\begin{proposition}
The hybrid zonotope $\mathcal{Z}_{\mathcal{H}} = \left\langle G_c, G_b, \mathbf{c}, A_c, A_b, \mathbf{b} \right\rangle$ can be equivalently expressed as $\Bar{\mathcal{Z}}_{\mathcal{H}} = \left\langle \bar{G}_c, \Bar{G}_b, \Bar{\mathbf{c}}, \Bar{A}_c, \Bar{A}_b, \Bar{\mathbf{b}} \right\rangle$ where the constraints on the factors have been modified such that
\begin{equation} \label{eq:hyb_zonotope_0_1}
\Bar{\mathcal{Z}}_\mathcal{H} = \left\{ \begin{bmatrix} \Bar{G}_c & \Bar{G}_b \end{bmatrix} 
\begin{bmatrix} \Bar{\bm{\xi}}_c \\ \Bar{\bm{\xi}}_b \end{bmatrix} + \Bar{\mathbf{c}} \middle| 
\begin{matrix}
\begin{bmatrix} \Bar{\bm{\xi}}_c \\ \Bar{\bm{\xi}}_b \end{bmatrix} \in [0,1]^{n_g} \times \{0,1\}^{n_b} \\
\begin{bmatrix} \Bar{A}_c & \Bar{A}_b \end{bmatrix} \begin{bmatrix} \Bar{\bm{\xi}}_c \\ \Bar{\bm{\xi}}_b \end{bmatrix} = \Bar{\mathbf{b}}
\end{matrix}
\right\} \;,
\end{equation}
using the matrix transformations
\begin{subequations}
\begin{align}
    &\Bar{G}_c = 2 G_c \;,\; \Bar{G}_b = 2 G_b \;,\; \Bar{\mathbf{c}} = \mathbf{c} - \begin{bmatrix} G_c & G_b \end{bmatrix} \mathbf{1} \:,\\
    &\Bar{A}_c = 2 A_c \;,\; \Bar{A}_b = 2 A_b \;,\; \Bar{\mathbf{b}} = \mathbf{b} + \begin{bmatrix} A_c & A_b \end{bmatrix} \mathbf{1} \:.
\end{align}
\end{subequations}
\begin{proof}
    Define $G = \begin{bmatrix} G_c & G_b \end{bmatrix}$, $A = \begin{bmatrix} A_c & A_b \end{bmatrix}$, and $\bm{\xi} = \begin{bmatrix} \bm{\xi}_c^T & \bm{\xi}_b^T \end{bmatrix}^T$. Then the hybrid zonotope generator and constraint equations are given by
\begin{subequations}
\begin{align}
    &\mathbf{y} = G \bm{\xi} + \mathbf{c} \;, \\
    &\mathbf{0} = A \bm{\xi} - \mathbf{b} \;.
\end{align}
\end{subequations}
Take $\bm{\xi} = 2 \Bar{\bm{\xi}} - \mathbf{1}$, then
\begin{subequations}
\begin{align}
   \mathbf{y} &= G (2\Bar{\bm{\xi}} - \mathbf{1}) + \mathbf{c} \;, \\
    &= 2G \Bar{\bm{\xi}} + \left( \mathbf{c} - G \mathbf{1} \right) \;,
\end{align}
\end{subequations}
and
\begin{subequations}
\begin{align}
    \mathbf{0} &= A (2 \Bar{\bm{\xi}} - \mathbf{1}) - \mathbf{b} \;, \\
    &= 2A \Bar{\bm{\xi}} - (\mathbf{b} + A\mathbf{1}) \;.
\end{align}
\end{subequations}
Thus $\Bar{G} = 2G$, $\Bar{\mathbf{c}} = \mathbf{c}-G \mathbf{1}$, $\Bar{A} = 2A$, and $\Bar{\mathbf{b}} = \mathbf{b} + A \mathbf{1}$ define a hybrid zonotope for which $\Bar{\bm{\xi}}_c \in [0,1]^{n_g}$ and $\Bar{\bm{\xi}}_b \in \{0,1\}^{n_b}$.
\end{proof}
\end{proposition}
Unless otherwise specified, hybrid zonotopes of the form in~\eqref{eq:hyb_zonotope_0_1} are used throughout the rest of this paper. Constrained zonotopes with modified factors $\Bar{\bm{\xi}}_c \in [0,1]^{n_g}$ are also used. These forms are used to improve the clarity of the proofs in Sec.~\ref{sec:conv-relax} and to facilitate comparisons to commercial mixed-integer solvers.

Two methods to construct hybrid zonotopes are considered. For the case where the polytopes $\mathcal{F}_i$ that compose the obstacle-free space do not have a grid structure, a hybrid zonotope is constructed from a union of vertex representation (V-rep) polytopes using \cite[Thm 5]{siefert2023reachability}. Inputs to this conversion are a vertex matrix $V=[\mathbf{v}_1,\dots,\mathbf{v}_{n_v}]\in\mathbb{R}^{n\times n_v}$ and an incidence matrix $M\in\mathbb{R}^{n_v \times n_F}$, $M_{ij} \in \{0,1\}$ such that
\begin{equation}
    \mathcal{F}_i=\left\{ V \bm{\lambda} \middle| \begin{array}{cc}
             \lambda_j \in \begin{cases} [0,\ 1], \; \text{if}\ j\in \{ k\ |\ M_{(k,i)}=1 \} \\ \{0\}, \; \:\:\:\, \text{if}\ j\in \{ k\ |\ M_{(k,i)}=0 \} \end{cases}  \\
             \mathbf{1}^T_{n_v} \bm{\lambda} = 1 
        \end{array} \right\}\;,
\end{equation}%
where $n_v$ is the total number of vertices in the set and $n_F$ is the number of polytopes in \eqref{eq:mpc-gen-obs-avoid}. The hybrid zonotope $\mathcal{F}$ is constructed equivalently to that in \cite[Thm 5]{siefert2023reachability}, altered for the form given in \eqref{eq:hyb_zonotope_0_1}, by
    \begin{subequations} \label{eq:vrep2hybzono_sets}
    \begin{align}
        &\mathcal{Q}=\left\langle\begin{bmatrix}
            I_{n_v} \\ 0
        \end{bmatrix},\begin{bmatrix}
            0 \\ I_{n_F}
        \end{bmatrix},\begin{bmatrix}
            \mathbf{0}_{n_v} \\ \mathbf{0}_{n_F}
        \end{bmatrix}, \begin{bmatrix}
            \mathbf{1}_{n_v}^T\\ \mathbf{0}^T
        \end{bmatrix}, \begin{bmatrix}
            \mathbf{0}^T \\ \mathbf{1}_{n_F}^T
        \end{bmatrix},\begin{bmatrix}
            1 \\ 1
        \end{bmatrix}\right\rangle\:,\\
        &\mathcal{H}=\{\mathbf{h}\in\mathbb{R}^{n_v}~\vert~ \mathbf{h}\leq\mathbf{0}\}\:,\\
        &\mathcal{D}=\mathcal{Q}\cap_{[\mathbf{I}_{n_v}~-M]}\mathcal{H}\:,\\
        &\mathcal{F}=\begin{bmatrix}
            V & 0
        \end{bmatrix}\mathcal{D}\:.
    \end{align}
    \end{subequations}
The $\mathcal{F}_i$ are constructed in V-rep from a polytopic obstacle description using the Hertel and Mehlhorn algorithm \cite{o1998computational}. Hybrid zonotopes of this form have dimensions $n_g = 2n_v$, $n_b = n_F$, $n_c = n_v+2$ where $n_v$ is the total number of vertices and $n_F$ is the number of polytopes composing $\mathcal{F}$.

When the obstacle-free space $\mathcal{F}$ is given as an OGM, it can be expressed as the hybrid zonotope
\begin{equation} \label{eq:hybzono-occ-grid}
    \mathcal{F} = \left\langle \begin{bmatrix} d_x & 0 \\ 0 & d_y \end{bmatrix}, \begin{bmatrix} \mathbf{c}_{b1} & \cdots & \mathbf{c}_{b n_F} \end{bmatrix}, \begin{bmatrix} -d_x/2 \\ -d_y/2 \end{bmatrix}, \mathbf{0}^T, \mathbf{1}^T, 1 \right\rangle \;,
\end{equation}
where $\mathbf{c}_{bi}$ gives the coordinates for the center of cell $i$ and $d_x$ and $d_y$ give the cell dimensions \cite{robbins2024efficient}. 
The hybrid zonotope dimensions for this case are $n_g = 2$, $n_b = n_F$, and $n_c = 1$.

In both \eqref{eq:vrep2hybzono_sets} and \eqref{eq:hybzono-occ-grid}, each binary factor $\xi_{b,i}$ corresponds to the convex region $\mathcal{F}_i$ such that fixing the binary factors to be 
\begin{equation}
\xi_{b,j} = \begin{cases} 1 \;,\; \text{if }i=j \\ 0 \;,\; \text{otherwise} \end{cases} \;,
\end{equation}
results in the following constrained zonotope representing the $i^{\text{th}}$ convex region:
\begin{equation}
    \mathcal{F}_i = \left\langle G_c, G_b \bm{\xi}_b + \mathbf{c}, A_c, \mathbf{b} - A_b \bm{\xi}_b \right\rangle \;.
\end{equation}
For a given point in $\mathcal{F}$, only one binary factor can be non-zero due to the choice constraints $\mathbf{1}^T\bm{\xi}_b = 1$ and $\xi_{b,i} \in \{0,1\} ,\; \forall i \in \{1, ..., n_F\}$. These properties are exploited within the mixed-integer solver described in Sec.~\ref{sec:mi_solver}.

\subsection{Convex Relaxation and Convex Hull}
Define the set $\mathcal{A}(\mathbf{x}_c, \mathbf{x}_i)$ to be mixed-integer convex such that it represents the feasible space of a mixed-integer convex program where $\mathbf{x}_c \in \mathcal{B} \subseteq \real^{n_c}$ and $\mathbf{x}_i \in \mathcal{C} \subseteq \integer^{n_i}$ are continuous and integer variables, respectively \cite{lubin2022mixed}. The interval hull, convex relaxation, and convex hull are defined as follows.
\begin{defn}
The interval hull of a set $\mathcal{D} = \{d_1, ..., d_{n_D}\}$ is given by $\mathit{IH}(\mathcal{D}) = [\min{(\mathcal{D})}, \max{(\mathcal{D})}]$.
\end{defn}
\begin{defn} \label{def:conv-relax}
The convex relaxation $\mathit{CR}(\mathcal{A}) \supseteq \mathcal{A}$ is given by replacing the constraint $\mathbf{x}_i \in \mathcal{C} \subseteq \integer^{n_i}$ with its interval hull such that $\mathbf{x}_i \in \mathit{IH}(\mathcal{C}) \subseteq \real^{n_i}$.
\end{defn}
\begin{defn} \label{def:conv-hull}
The convex hull $\mathit{CH}(\mathcal{A})$ is the tightest convex set such that $\mathcal{A} \subseteq \mathit{CH}(\mathcal{A})$ and is given by the set of all convex combinations of points $\mathbf{x}_i$ in $\mathcal{A}$:
\begin{equation}
    \mathit{CH}(\mathcal{A}) = \left\{ \sum_i \lambda_i \mathbf{x}_i \middle| \mathbf{x}_i \in \mathcal{A} \;,\; \sum_i \lambda_i = 1 \;,\; \lambda_i \geq 0 \right\} \;.
\end{equation}
\end{defn}

\section{Convex Relaxations of Hybrid Zonotopes} \label{sec:conv-relax}

In mixed-integer programming, tight convex relaxations reduce the number of branch-and-bound iterations by producing less conservative objective function lower bounds that result in earlier pruning of sub-optimal branches \cite{belotti2016handling, bonami2015mathematical, frangioni2006perspective, karamanov2006branch, bertsimas2021unified, marcucci2024shortest}. Most mixed-integer AV motion planning formulations use the Big-M method \cite{ioan2021mixed} which is known to have weak convex relaxations \cite{bonami2015mathematical}.
This section shows that hybrid zonotopes of the forms given in \eqref{eq:vrep2hybzono_sets} and \eqref{eq:hybzono-occ-grid} have the property that their convex relaxation is their convex hull.

For a hybrid zonotope $\mathcal{F} = \left\langle G_c, G_b, \mathbf{c}, A_c, A_b, \mathbf{b} \right\rangle$, the convex relaxation is given as $\mathit{CR}(\mathcal{F}) = \left\langle [G_c \; G_b], [], \mathbf{c}, [A_c \; A_b], [], \mathbf{b} \right\rangle$. Theorem~\ref{thm:conv-relax-vrep} shows that for hybrid zonotopes constructed from V-rep polytopes using \cite[Thm 5]{siefert2023reachability}, the convex relaxation is the convex hull.

\begin{theorem} \label{thm:conv-relax-vrep}
    Given a hybrid zonotope $\mathcal{F} = \left\langle G_c, G_b, \mathbf{c}, A_c, A_b, \mathbf{b} \right\rangle$ constructed with \eqref{eq:vrep2hybzono_sets}, $\mathit{CR}(\mathcal{F}) = \mathit{CH}(\mathcal{F})$.
\end{theorem}
\begin{proof}
    $\mathit{CR}(\mathcal{F})$ is convex and $\mathcal{F} \subseteq \mathit{CR}(\mathcal{F})$, therefore $\mathit{CH}(\mathcal{F})\subseteq \mathit{CR}(\mathcal{F})$. 

    From \eqref{eq:vrep2hybzono_sets},
    \begin{subequations} \label{eq:vrep2hybzono_sets_rewritten}
    \begin{align}
        &\mathcal{Q} = \left\{ \begin{bmatrix}
            \bm{\xi}_c \\
            \bm{\xi}_b 
        \end{bmatrix} \middle| \begin{matrix} \sum_i \xi_{ci} = 1 \;,\; \sum_i \xi_{ci} \in [0,1] \\ \sum_i \xi_{bi} = 1 \;,\; \xi_{bi} \in \{0,1\} \end{matrix} \right\} \;, \label{eq:Q_v2} \\
        &\mathcal{D} = \left\{\begin{bmatrix}
            \bm{\xi}_c \\
            \bm{\xi}_b 
        \end{bmatrix} \in Q \middle| \begin{bmatrix} I & -M \end{bmatrix} \begin{bmatrix}
            \bm{\xi}_c \\
            \bm{\xi}_b 
        \end{bmatrix} \in \mathcal{H} \right\} \;, \\
        &\mathcal{F} = \left\{ \begin{bmatrix} V & 0 \end{bmatrix} \begin{bmatrix}
            \bm{\xi}_c \\
            \bm{\xi}_b 
        \end{bmatrix} \middle| \begin{matrix} \sum_i \xi_{ci} = 1 \;,\; \xi_{ci} \in [0,1] \\ \sum_i \xi_{bi} = 1 \;,\; \xi_{bi} \in \{0,1\} \\
        \begin{bmatrix} I_{n_v} & -M \end{bmatrix} \begin{bmatrix} \bm{\xi}_c \\ \bm{\xi}_b \end{bmatrix} \in \mathcal{H}
        \end{matrix} \right\} \;. \label{eq:Z-vrep-gen}
    \end{align}
    \end{subequations}
    The convex relaxation of $\mathcal{F}$ is then given as
    \begin{subequations}
    \begin{align} \label{eq:z_vrep_CR}
        \mathit{CR}(\mathcal{F}) &= \left\{ \sum_i \xi_{ci} \mathbf{v}_i \middle| \begin{matrix} \sum_i \xi_{ci} = 1 \;,\; \xi_{ci} \in [0,1] \\
        \sum_i \xi_{bi} = 1 \;,\; \xi_{bi} \in [0,1] \\
        \begin{bmatrix} I_{n_v} & -M \end{bmatrix} \begin{bmatrix} \bm{\xi}_c \\ \bm{\xi}_b \end{bmatrix} \in \mathcal{H} \end{matrix}\right\} \;,\\
        &\subseteq \left\{ \sum_i \xi_{ci} \mathbf{v}_i \middle| \begin{matrix} \sum_i \xi_{ci} = 1 \;,\; \xi_{ci} \in [0,1] \end{matrix}\right\} \;.
    \end{align}
    \end{subequations}
   Thus $\mathit{CR}(\mathcal{F})\subseteq \mathit{CH}(\mathcal{F})$ and $\mathit{CR}(\mathcal{F}) = \mathit{CH}(\mathcal{F})$.    
\end{proof}

For OGMs, this paper constructs hybrid zonotope representations using \eqref{eq:hybzono-occ-grid}. The following lemmas and theorem can be used to show that the convex relaxation of a hybrid zonotope in this form is its convex hull.

\begin{lemma} \label{lem:conv_hull_pts}
    Given a hybrid zonotope $\mathcal{F}$ where
    \begin{equation} \label{eq:hybzono_pt_collection}
        \mathcal{F} = \left\langle [], \begin{bmatrix} \mathbf{c}_{b1} & \cdots & \mathbf{c}_{bn_F} \end{bmatrix}, \mathbf{0}, [], \mathbf{1}^T, 1 \right\rangle \;,
    \end{equation}
    $\mathit{CR}(\mathcal{F}) = \mathit{CH}(\mathcal{F})$.
    \begin{proof}
        The set $\mathcal{F}$ is the union of points $\mathbf{c}_{bi}$,
        \begin{subequations}
        \begin{align}
            \mathcal{F} &= \sum_i \lambda_i \mathbf{c}_{bi},\; \sum_i \lambda_i = 1,\, \lambda_i \in \{0, 1\} \;, \\
            &= \bigcup_i \mathbf{c}_{bi} \;.
        \end{align}
        \end{subequations}
        Taking the convex relaxation gives
        \begin{equation}
            \mathit{CR}(\mathcal{F}) = \sum_i \lambda_i \mathbf{c}_{bi},\; \sum_i \lambda_i = 1,\, \lambda_i \in [0, 1] \;, 
        \end{equation}
        which is the convex hull of the union of points $\mathbf{c}_{bi}$ by definition, thus $\mathit{CR}(\mathcal{F}) = \mathit{CH}(\mathcal{F})$.  
    \end{proof}
\end{lemma}

\begin{lemma} \label{lem:stacked_hz_conv_hull}
    Given hybrid zonotopes
    \begin{subequations} \label{eq:hybzono_relaxations_2_hybzonos}
    \begin{align}
        &\mathcal{F}_1 = \langle G_{c1}, G_{b1}, \mathbf{c}_1, A_{c1}, A_{b1}, \mathbf{b}_1 \rangle \;, \\
        &\mathcal{F}_2 = \langle G_{c2}, G_{b2}, \mathbf{c}_2, A_{c2}, A_{b2}, \mathbf{b}_2 \rangle \;,
    \end{align}
    \end{subequations}
    and their Minkowski sum \cite[Prop. 7]{bird2023hybrid}
    \begin{multline} \label{eq:hybzono_relaxations_Z_Z1_plus_Z2}
        \mathcal{F} = \mathcal{F}_1 \oplus \mathcal{F}_2 = \left\langle \begin{bmatrix} G_{c1} & G_{c2}\end{bmatrix}, \begin{bmatrix} G_{b1} & G_{b2}\end{bmatrix}, \mathbf{c}_1 + \mathbf{c}_2 \right., \\ \left. \begin{bmatrix} A_{c1} & 0 \\ 0 & A_{c2} \end{bmatrix}, \begin{bmatrix} A_{b1} & 0 \\ 0 & A_{b2} \end{bmatrix}, \begin{bmatrix} \textbf{b}_1 \\ \textbf{b}_2 \end{bmatrix} \right\rangle \;,
    \end{multline}
    if $\mathit{CR}(\mathcal{F}_1)=\mathit{CH}(\mathcal{F}_1)$ and $\mathit{CR}(\mathcal{F}_2)=\mathit{CH}(\mathcal{F}_2)$, then $\mathit{CR}(\mathcal{F}) = \mathit{CH}(\mathcal{F})$.
    \begin{proof}
       The convex relaxation of $\mathcal{F}$ is given as 
       \begin{multline}
           \mathit{CR}(\mathcal{F}) = \left\langle \begin{bmatrix}
               G_{c1} & G_{c2} & G_{b1} & G_{b2}
           \end{bmatrix}, [], \mathbf{c}_1 + \mathbf{c}_2,  \right. \\
           \left.
           \begin{bmatrix}
               A_{c1} & 0 & A_{b1} & 0 \\
               0 & A_{c2} & 0 & A_{b2} 
           \end{bmatrix}, [], \begin{bmatrix}
               \mathbf{b}_1 \\ \mathbf{b}_2
           \end{bmatrix}   \right\rangle \;.
       \end{multline}
        By \cite[Prop. 7]{bird2023hybrid}, $\mathit{CR}(\mathcal{F}) = \mathit{CR}(\mathcal{F}_1) \oplus \mathit{CR}(\mathcal{F}_2)$. If $\mathit{CR}(\mathcal{F}_1) = \mathit{CH}(\mathcal{F}_1)$ and $\mathit{CR}(\mathcal{F}_2) = \mathit{CH}(\mathcal{F}_2)$, then $\mathit{CR}(\mathcal{F}) = \mathit{CH}(\mathcal{F}_1 \oplus \mathcal{F}_2) = \mathit{CH}(\mathcal{F})$ using the fact that the Minkowski sum and convex hull commute.  
    \end{proof}
\end{lemma}

\begin{theorem} \label{thm:conv-relax-occ-grid}
    Given a hybrid zonotope $\mathcal{F}$ in the form of \eqref{eq:hybzono-occ-grid}, $\mathit{CR}(\mathcal{F}) = \mathit{CH}(\mathcal{F})$.
    \begin{proof}
        By \cite[Prop. 7]{bird2023hybrid}, $\mathcal{F} = \mathcal{F}_1 \oplus \mathcal{F}_2$ where 
        \begin{equation}
            \mathcal{F}_1 =\left\langle \begin{bmatrix} d_x & 0 \\ 0 & d_y \end{bmatrix}, [], \begin{bmatrix}
             -d_x/2 \\ -d_y/2
            \end{bmatrix}, [], [], [] \right\rangle \;,
        \end{equation}
        and $\mathcal{F}_2$ is given by \eqref{eq:hybzono_pt_collection}. $\mathcal{F}_1$ is convex because it is a zonotope, therefore $\mathit{CR}(\mathcal{F}_1) = \mathit{CH}(\mathcal{F}_1)$. By Lemma~\ref{lem:conv_hull_pts}, $\mathit{CR}(\mathcal{F}_2) = \mathit{CH}(\mathcal{F}_2)$. Lemma~\ref{lem:stacked_hz_conv_hull} applies and $\mathit{CR}(\mathcal{F}) = \mathit{CH}(\mathcal{F})$.
    \end{proof}
\end{theorem}

\section{Solution of Mixed-Integer MPC Problems Using Hybrid Zonotope Constraints} \label{sec:solver}
This section presents a multi-stage MIQP formulation and solution strategy for MPC motion planning problems where hybrid zonotopes are used to represent non-convex constraints on the system output.
The high-level MIQP structure is given in Sec.~\ref{sec:miqp}. The MIQP is built up starting with a convex problem formulation 
in Sec.~\ref{sec:baseline_prob_form}, and hybrid zonotope constraints are added in Sec.~\ref{sec:problem-formulation-hybzono}. The particular structure of this formulation is exploited within a branch-and-bound mixed-integer solver (Sec.~\ref{sec:mi_solver}), and a solver for QP sub-problems generated by the mixed-integer solver (Sec.~\ref{sec:qp_solver}).

\subsection{Problem Formulation} \label{sec:problem-formulation}

\subsubsection{Multi-Stage MIQP} \label{sec:miqp}
Based on Assumptions~\ref{as:X-U-XT} and \ref{as:F-polytope}, Eq.~\eqref{eq:mpc-gen} can be written as the multi-stage MIQP
\begin{subequations} \label{eq:miqp-multistage}
\begin{align} 
&\mathbf{z}^* = \argmin_{\mathbf{z}} \sum_{k=0}^{N} \frac{1}{2} \mathbf{z}_k^T P_k \mathbf{z}_k + \mathbf{q}_k^T \mathbf{z}_k\;, \label{eq:miqp-multistage-cost} \\
&\text{s.t.} \; \mathbf{0} = C_k \mathbf{z}_k + D_{k+1} \mathbf{z}_{k+1} + \mathbf{c}_k,\; \forall k \in \mathcal{K}\;, \label{eq:miqp-multistage-eqcons}\\
&\phantom{\text{s.t.}} \; G_k \mathbf{z}_k \leq \mathbf{w}_k,\; \forall k \in \mathcal{K} \cup N \;,\label{eq:miqp-multistage-ineqcons}
\end{align}
\end{subequations}
with $\mathbf{z}_k = \begin{bmatrix} \mathbf{x}_k^T & \mathbf{u}_k^T & \bm{\alpha}_{k}^T \end{bmatrix}^T$. The vector $\bm{\alpha}_k \in \real^{n_{\mathit{\alpha c} k}} \times \integer^{n_{\mathit{\alpha b} k}}$ denotes variables used to define constraints. The binary variables in $\bm{\alpha}_k$ are required to be in the form of a choice constraint where each binary variable corresponds to a convex region of the free space. This is the case for hybrid zonotopes constructed from V-rep polytopes and for hybrid zonotope representations of OGMs. Unions of H-rep polytopes can also be formulated with choice constraints on the binary variables using the Big-M method~\cite{robbins2024efficient}. The constraints \eqref{eq:mpc-gen-state-input-cons} and \eqref{eq:mpc-gen-obs-avoid} can be softened as is often required for practical MPC implementations \cite{borrelli2017predictive}. In this case, $\bm{\alpha}_k$ will contain slack variables as discussed in Sections~\ref{sec:baseline_prob_form} and \ref{sec:problem-formulation-hybzono}.

\subsubsection{Convex Problem Formulation} \label{sec:baseline_prob_form}
The convex problem formulation corresponds to \eqref{eq:mpc-gen} without non-convex constraints or region-dependent costs (i.e., without \eqref{eq:mpc-gen-obs-avoid} and \eqref{eq:mpc-gen-reg-dep-costs} and with $q^r(\mathbf{y}_k)$ removed from \eqref{eq:mpc-gen-cost}). Non-convex constraints and region-dependent costs will be treated in Sec.~\ref{sec:problem-formulation-hybzono}.

The polytopic state, terminal state, and input constraint sets $\mathcal{X}$, $\mathcal{X}_N$, and $\mathcal{U}$, respectively, are each constructed using constrained zonotopes for a linear map of the corresponding signal and box constraints for the signal itself. These sets are given as
\begin{subequations} \label{eq:state-input-cons-conzono}
\begin{align}
&\mathcal{X} = \left\{\mathbf{x} \middle| H_x \mathbf{x} \in \mathcal{Z}_{cx},\; \mathbf{x} \in [\mathbf{x}^l, \mathbf{x}^u] \right\} \;, \\
&\mathcal{X}_N = \left\{\mathbf{x} \middle| H_{xN} \mathbf{x} \in \mathcal{Z}_{cxN},\; \mathbf{x} \in [\mathbf{x}_N^l, \mathbf{x}_N^u] \right\} \;, \\
&\mathcal{U} = \left\{\mathbf{u} \middle| H_u \mathbf{u} \in \mathcal{Z}_{cu},\; \mathbf{u} \in [\mathbf{u}^l, \mathbf{u}^u] \right\} \;,
\end{align}
\end{subequations}
where $\mathcal{Z}_{cx} = \left\langle G_{cx}, \mathbf{c}_{cx}, A_{cx}, \mathbf{b}_{cx} \right\rangle$, $\mathcal{Z}_{cxN} = \left\langle G_{cxN}, \mathbf{c}_{cxN}, A_{cxN}, \mathbf{b}_{cxN} \right\rangle$, and $\mathcal{Z}_{cu} = \left\langle G_{cu}, \mathbf{c}_{cu}, A_{cu}, \mathbf{b}_{cu} \right\rangle$ are constrained zonotopes. The matrices $H_x$, $H_{x N}$, and $H_u$ extract the signals which participate in the constrained zonotope constraints.
It will be shown that constraints of this form ensure that the inequality constraint matrices $G_k$, $k \in \{0, ..., N\}$ are in box constraint form and are full column rank. This structure is exploited in the QP solver presented in Sec.~\ref{sec:qp_solver}. H-rep polytopic constraints, by comparison, do not in general have this structure \cite{domahidi2012efficient}.

The constrained zonotope constraints in \eqref{eq:state-input-cons-conzono} can be softened while the box constraints must be hard to ensure the required structure for $G_k$. In practice, the box constraint bounds on a variable $z_i$ may be set to large values when approximating $z_i \in [-\infty, \infty]$ or when $z_i$ is subject to constrained zonotope or hybrid zonotope constraints. The form in \eqref{eq:state-input-cons-conzono} is flexible such that hard box constraints can be used directly where applicable without requiring redundant definitions of those constraints within the constrained zonotopes.

Expressing \eqref{eq:mpc-gen-state-input-cons} using \eqref{eq:state-input-cons-conzono} requires the introduction of additional optimization variables $\bm{\xi}_{cx}$ and $\bm{\xi}_{cu}$ to the convex problem formulation. The vector of optimization variables then becomes $\mathbf{z}_{k}^{c} = \begin{bmatrix}
    \mathbf{x}_k^T & \mathbf{u}_k^T & \bm{\xi}_{cx k}^T & \bm{\xi}_{cu k}^T & \bm{\sigma}_{cx k}^T & \bm{\sigma}_{cu k}^T
\end{bmatrix}^T$ for $k \in \mathcal{K}$ where $\bm{\sigma}_{cx k}$ and $\bm{\sigma}_{cu k}$ denote optional slack variables corresponding to the constrained zonotope state and input constraints, respectively. Variables corresponding to the control inputs disappear at $k=N$ such that $\mathbf{z}_N^{c} = \begin{bmatrix}
    \mathbf{x}_N^T & \bm{\xi}_{cx N}^T & \bm{\sigma}_{cx N}^T
\end{bmatrix}^T$. To ensure that the inequality constraint matrices $G_k$ are full column rank, the slack variables are subject to box constraints $\bm{\sigma}_{cxk} \in [-\bm{\sigma}_{cx}^u, \bm{\sigma}_{cx}^u]$ and $\bm{\sigma}_{cuk} \in [-\bm{\sigma}_{cu}^u, \bm{\sigma}_{cu}^u]$ for $k \in \{0, ..., N-1\}$ and $\bm{\sigma}_{cxN} \in [-\bm{\sigma}_{cxN}^u, \bm{\sigma}_{cxN}^u]$.

For $k \in \{0, ..., N-1\}$, the cost function \eqref{eq:miqp-multistage-cost} matrices and vectors within the convex problem are 
\begin{subequations} \label{eq:conv-prob-cost}
\begin{align}
&P_k^{c} = \text{blkdiag} \left( \begin{bmatrix}
    Q_k & R_k & 0 & 0 & W_{cxk} & W_{cuk}
\end{bmatrix} \right) \;, \\
&\mathbf{q}_k^{c} = \begin{bmatrix}
    -(Q_k \mathbf{x}_k^r)^T & \mathbf{0}^T & \mathbf{0}^T & \mathbf{0}^T & \mathbf{0}^T & \mathbf{0}^T
\end{bmatrix}^T \;.
\end{align}
\end{subequations}
The matrices $W_{cxk}$ and $W_{cuk}$ are diagonal and define the cost for nonzero values of the slack variables. For $k=N$, \eqref{eq:conv-prob-cost} is modified to remove costs corresponding to the control inputs and $Q_N$ and $W_{cx N}$ are substituted for $Q_k$ and $W_{cx k}$. The diagonal matrix $W_{cx N}$ assigns costs to the terminal state slack variables $\bm{\sigma}_{cxN}$. 

For $k \in \{1, ..., N-2\}$, the matrices and vectors that define the state and input equality constraints \eqref{eq:miqp-multistage-eqcons} in the convex problem formulation are
\begin{subequations} \label{eq:conv-prob-eq-cons}
\begin{align}
&C_k^{c} = \begin{bmatrix}
    A & B & 0 & 0 & 0 & 0 \\
    H_x & 0 & -G_{cx} & 0 & I & 0 \\
    0 & H_u & 0 & -G_{cu} & 0 & I \\
    0 & 0 & A_{cx} & 0 & 0 & 0 \\
    0 & 0 & 0 & A_{cu} & 0 & 0
\end{bmatrix} \;, \\
&D_{k+1}^{c} = \begin{bmatrix}
    -I & 0 & 0 & 0 & 0 & 0 \\
    0 & 0 & 0 & 0 & 0 & 0 \\
    0 & 0 & 0 & 0 & 0 & 0 \\
    0 & 0 & 0 & 0 & 0 & 0 \\
    0 & 0 & 0 & 0 & 0 & 0
\end{bmatrix} \;, \\
&\mathbf{c}_k^{c} = \begin{bmatrix}
    \mathbf{0}^T & -\mathbf{c}_{cx}^T & -\mathbf{c}_{cu}^T & -\mathbf{b}_{cx}^T & -\mathbf{b}_{cu}^T
\end{bmatrix}^T \;.
\end{align}
\end{subequations}
For $k=0$, \eqref{eq:conv-prob-eq-cons} is modified to add the constraint $-I \mathbf{x}_k + \mathbf{x}_0 = \mathbf{0}$. For the case that the initial control input is fixed, as is common to account for computation times in real-time MPC formulations, the constraint $-I \mathbf{u}_k + \mathbf{u}_0 = \mathbf{0}$ is additionally included. 
The matrices and vector used in \eqref{eq:miqp-multistage-eqcons} at $k=N-1$ and $k=N$ are 
\begin{subequations} \label{eq:conv-prob-eq-cons-term}
\begin{align}
&C_{N-1}^{c} = \begin{bmatrix}
    A & B & 0 & 0 & 0 & 0 \\
    H_x & 0 & -G_{cx} & 0 & I & 0 \\
    0 & H_u & 0 & -G_{cu} & 0 & I \\
    0 & 0 & A_{cx} & 0 & 0 & 0 \\
    0 & 0 & 0 & A_{cu} & 0 & 0 \\
    0 & 0 & 0 & 0 & 0 & 0 \\
    0 & 0 & 0 & 0 & 0 & 0
\end{bmatrix} \;, \\
&D_{N}^{c} = \begin{bmatrix}
    -I & 0 & 0 \\
    0 & 0 & 0 \\
    0 & 0 & 0 \\
    0 & 0 & 0 \\
    0 & 0 & 0 \\
    H_{x N} & -G_{cx N} & I \\
    0 & A_{cx N} & 0
\end{bmatrix} \;, \\
&\mathbf{c}_{N-1}^{c} = \left[ \begin{matrix}
    \mathbf{0}^T & -\mathbf{c}_{cx}^T & -\mathbf{c}_{cu}^T & -\mathbf{b}_{cx}^T & -\mathbf{b}_{cu}^T \end{matrix} \right. \nonumber \\
    & \qquad \qquad \qquad \qquad \qquad \qquad \; \left. \begin{matrix} -\mathbf{c}_{cx N}^T & -\mathbf{b}_{cx N}^T \end{matrix} \right]^T \;.
\end{align}
\end{subequations}

For $k \in \{0, ..., N-1\}$, the convex problem matrices and vectors used in the inequality constraints \eqref{eq:miqp-multistage-ineqcons} are
\begin{subequations} \label{eq:conv-prob-ineq-cons}
\begin{align}
&G_k^{c} = \text{blkdiag}\left(\begin{bmatrix} -I \\ I \end{bmatrix}, \begin{bmatrix} -I \\ I \end{bmatrix}, \begin{bmatrix} -I \\ I \end{bmatrix}, \right. \nonumber \\
& \qquad \qquad \qquad \; \, \left. \begin{bmatrix} -I \\ I \end{bmatrix}, \begin{bmatrix} -I \\ I \end{bmatrix}, \begin{bmatrix} -I \\ I \end{bmatrix} \right) \; , \\
&\mathbf{w}_k^{c} = \left[ \begin{matrix}
    -(\mathbf{x}^l)^T & (\mathbf{x}^u)^T & -(\mathbf{u}^l)^T & (\mathbf{u}^u)^T & \mathbf{0}^T & \mathbf{1}^T \end{matrix} \right. \nonumber \\
    & \qquad\quad \left. \begin{matrix} \mathbf{0}^T & \mathbf{1}^T & (\bm{\sigma}_{cx}^u)^T & (\bm{\sigma}_{cx}^u)^T & (\bm{\sigma}_{cu}^u)^T & (\bm{\sigma}_{cu}^u)^T
\end{matrix} \right]^T \;.
\end{align}
\end{subequations}
For $k=N$, \eqref{eq:conv-prob-ineq-cons} is modified by removing constraints on the inputs and substituting $\mathbf{x}_N^l$, $\mathbf{x}_N^u$, and $\bm{\sigma}_{cx N}^u$ for $\mathbf{x}^l$, $\mathbf{x}^u$, and $\bm{\sigma}_{cx}^u$, respectively. Clearly, the inequality constraint matrices $G_k^c$ are in box constraint form for all $k \in \{0, ..., N\}$.

\subsubsection{Hybrid Zonotope Constraints} \label{sec:problem-formulation-hybzono}
The hybrid zonotope constraints are built on top of the convex problem formulation. Additional optimization variables are introduced for the continuous and binary factors $\bm{\xi}_{ck} \in [0,1]^{n_g}$ and $\bm{\xi}_{bk} \in \{0,1\}^{n_b}$, respectively, and for the slack variables $\bm{\sigma}_{k}$ in the case of softened hybrid zonotope constraints. The full vector of optimization variables corresponding for time step $k \in \mathcal{K} \cup N$, is given by $\mathbf{z}_k = \begin{bmatrix} (\mathbf{z}_k^c)^T & \bm{\xi}_{ck}^T & \bm{\xi}_{bk}^T & \bm{\sigma}_k^T \end{bmatrix}^T$. For all $k \in \mathcal{K} \cup N$, the cost function matrices and vectors in \eqref{eq:miqp-multistage-cost} are
\begin{subequations} \label{eq:cost_fcn_Pk_qk}
\begin{align}
&P_k = \text{blkdiag} \left( \begin{bmatrix}
    P_k^{c} & 0 & 0 & W_k 
\end{bmatrix} \right) \;, \label{eq:cost_fcn_Pk} \\
&\mathbf{q}_k = \begin{bmatrix}
    (\mathbf{q}_k^{c})^T & \mathbf{0}^T & (\mathbf{q}^{r})^T & \mathbf{0}^T
\end{bmatrix}^T \;, \label{eq:cost_fcn_qk}
\end{align}
\end{subequations}
where $W_k$ is a diagonal quadratic cost matrix for the slack variables $\bm{\sigma}_k$, and $\mathbf{q}^r = \begin{bmatrix} q_{1}^r & \cdots & q_{n_F}^r \end{bmatrix}$ is the vector of linear costs associated with the region selections.

Using the hybrid zonotope $\mathcal{Z}_H = \langle G_c^{hz}, G_b^{hz}, \mathbf{c}^{hz}, A_c^{hz}, A_b^{hz}, \mathbf{b}^{hz} \rangle$ such that $H \mathbf{x}_k \in \mathcal{Z}_H$, the matrices and vectors for the full problem equality constraints \eqref{eq:miqp-multistage-eqcons} are given for $k \in \{0, ..., N-2\}$ by 
\begin{subequations} \label{eq:C_k-hybzono}
\begin{align}
&C_k = \begin{bmatrix} C_k^{c} & 0 & 0 & 0 \\
\begin{bmatrix} H & 0 \end{bmatrix} & -G_c^{hz} & -G_b^{hz} & I \\
0 & A_c^{hz} & A_b^{hz} & 0
\end{bmatrix} \;, \\
&D_{k+1} = \begin{bmatrix}
D_k^{c} & 0 & 0 & 0 \\
0 & 0 & 0 & 0 \\ 
0 & 0 & 0 & 0
\end{bmatrix} \;, \\
&\mathbf{c}_k = \begin{bmatrix}
    (\mathbf{c}_k^{c})^T & -(\mathbf{c}^{hz})^T & -(\mathbf{b}^{hz})^T
\end{bmatrix}^T \;.
\end{align}
\end{subequations}
The terminal equality constraint matrices and vectors are given by
\begin{subequations}
\begin{align}
    &C_{N-1} = \begin{bmatrix} C_{N-1}^{c} & 0 & 0 & 0 \\
\begin{bmatrix} H & 0 \end{bmatrix} & -G_c^{hz} & -G_b^{hz} & I \\
0 & A_c^{hz} & A_b^{hz} & 0 \\ 
0 & 0 & 0 & 0 \\
0 & 0 & 0 & 0
\end{bmatrix} \;, \\
&D_N = \begin{bmatrix}
D_N^{c} & 0 & 0 & 0 \\
0 & 0 & 0 & 0 \\
0 & 0 & 0 & 0 \\
0 & -G_c^{hz} & -G_b^{hz} & I \\
0 & A_c^{hz} & A_b^{hz} & 0
\end{bmatrix} \;, \\
&\mathbf{c}_{N-1} = \left[ \begin{matrix}
    (\mathbf{c}_{N-1}^{c})^T & -(\mathbf{c}^{hz})^T & -(\mathbf{b}^{hz})^T \end{matrix} \right. \nonumber \\
& \qquad \qquad \qquad \qquad \; \, \left. \begin{matrix}-(\mathbf{c}^{hz})^T & -(\mathbf{b}^{hz})^T
\end{matrix} \right]^T \;.
\end{align}
\end{subequations}

For all $k \in \{0, ..., N\}$, the inequality constraints for the full problem formulation are constructed by appending the hybrid zonotope box constraints slack variable bounds to $G_k^{c}$ as
\begin{subequations}
\begin{align}
&G_k = \begin{bmatrix}
    G_k^{c} & 0 & 0 & 0 \\
    0 & -I & 0 & 0 \\
    0 & I & 0 & 0 \\
    0 & 0 & -I & 0 \\
    0 & 0 & I & 0 \\
    0 & 0 & 0 & -I \\
    0 & 0 & 0 & I
\end{bmatrix} \;, \label{eq:ineq_cons_Gk} \\
&\mathbf{w}_k = \begin{bmatrix}
    \mathbf{w}_k^{c T} & \mathbf{0}^T & \mathbf{1}^T & \mathbf{0}^T & \mathbf{1}^T & -(\bm{\sigma}^u)^T & (\bm{\sigma}^u)T
\end{bmatrix}^T \;.
\end{align}
\end{subequations}
The slack variables are subject to the box constraint $\bm{\sigma} \in [-\bm{\sigma}^u, \bm{\sigma}^u]$. The inequality constraint matrices $G_k$ for the full problem formulation can be seen to be in box constraint form.

\subsection{Mixed-Integer Solver} \label{sec:mi_solver}
This section describes a mixed-integer solver designed to solve the multi-stage MIQP described in Sec.~\ref{sec:problem-formulation}. This solver utilizes a branch-and-bound solution strategy with support for concurrent exploration of multiple branches via multithreading. Throughout this section, the shorthand notation $\mathbf{r}_k = \mathbf{f}_r(\bm{\xi}_{b k})$ is used to represent a set of region indices at time step $k$ where index $i$ corresponds to region $\mathcal{F}_i$ in \eqref{eq:mpc-gen-obs-avoid}. The collection of $\mathbf{r}_k$ across all time steps is denoted as $\bar{\mathbf{r}} = \{\mathbf{r}_0, ..., \mathbf{r}_N\}$.

\subsubsection{Branch-and-Bound Overview}

Branch-and-bound algorithms solve mixed-integer convex programs by exploring a tree of convex programs for which integer variables in the original program are relaxed to continuous variables. The convex programs generated by the mixed-integer solver are referred to as ``nodes".  The top-level node in the tree, referred to as the ``root node'', is the convex relaxation of the original mixed-integer program. From a given node, ``children nodes" are constructed by fixing integer variables or applying additional constraints (e.g., $z_i \in [0,1] \rightarrow z_i=1$). This operation is known as ``branching''. 

The optimal objective of a node is a lower bound on the objective of that node's children as all optimization variables $\mathbf{z}$ that are feasible for the child node are also feasible for its parent. Branch-and-bound algorithms leverage the lower bounds by eliminating nodes
for which the cost lower bound is greater than the cost of the best identified feasible solution, called the ``incumbent''. The algorithm continues until there are either no nodes left to be explored or the objective lower bound and the objective of the incumbent are within some convergence tolerance.

\subsubsection{Reachability Constraints}
The mixed-integer solver described in this paper requires the multistage structure given in \eqref{eq:miqp-multistage}. Additionally, each binary factor $\xi_{bk, i}$ must correspond to a convex region of the obstacle-free space $\mathcal{F}_i$ and only one binary factor can be non-zero for each time step $k$. These properties hold for the hybrid zonotopes used in this paper as discussed in Sec.~\ref{sec:problem-formulation-hybzono}. To efficiently solve problems with this structure, the mixed integer solver presented in this paper uses a notion of ``reachability'' between region selections across time steps, originally introduced in \cite{robbins2024efficient}.
We define this notion of reachability as follows.
\begin{defn} \label{def:reachability}
A region $r_i$ is $k_n$ steps reachable from a point $\mathbf{y}_0$ (or another region $r_j$) if $\lfloor (d_r/d_{max}) \rfloor \geq k_n$, where $d_r$ is the length of the shortest line segment from $\mathbf{y}_0$ (or $r_j$) to $r_i$ and $d_{max}$ is the maximum distance that the AV can travel in a single time step of the MPC problem. 
\end{defn}

This is illustrated in Fig.~\ref{fig:constraint_region_reachability}, where region 3 is 2 steps reachable from the AV's current position (indicated by the green diamond), and region 6 is 3 steps reachable from region 9 and vice-versa. 

Algorithm~\ref{alg:r2r-p2r-reachability} returns the indices of the regions that can be reached within $k_n$ time steps of a point or region. Lines~\ref{algline:qp-region} and \ref{algline:qp-point} require the solution of a QP. 
Mehrotra's primal-dual interior point method for QPs is used for this purpose \cite{mehrotra1992implementation, borrelli2017predictive}.
In this implementation, the ``region" case QPs are solved offline because the obstacle map is assumed fully known and time invariant, while the ``point" case QPs are solved before the start of every MPC iteration. Reachability is only re-computed for regions that were reachable within the MPC horizon in the previous iteration and regions that are 1 step reachable from those regions.
These calculations are used to produce a point case lookup table $R_p(k_n)$ and a region case lookup table $R_r(k_n, r)$. The point case lookup table returns region indices $\mathbf{r}$ that are $k_n$ steps reachable from the initial position $H \mathbf{x}_0$. The region case lookup table returns region indices $\mathbf{r}$ that are $k_n$ steps reachable from the region $r$.

\begin{figure}[tb]
    \centering
    \includegraphics[width=\linewidth]{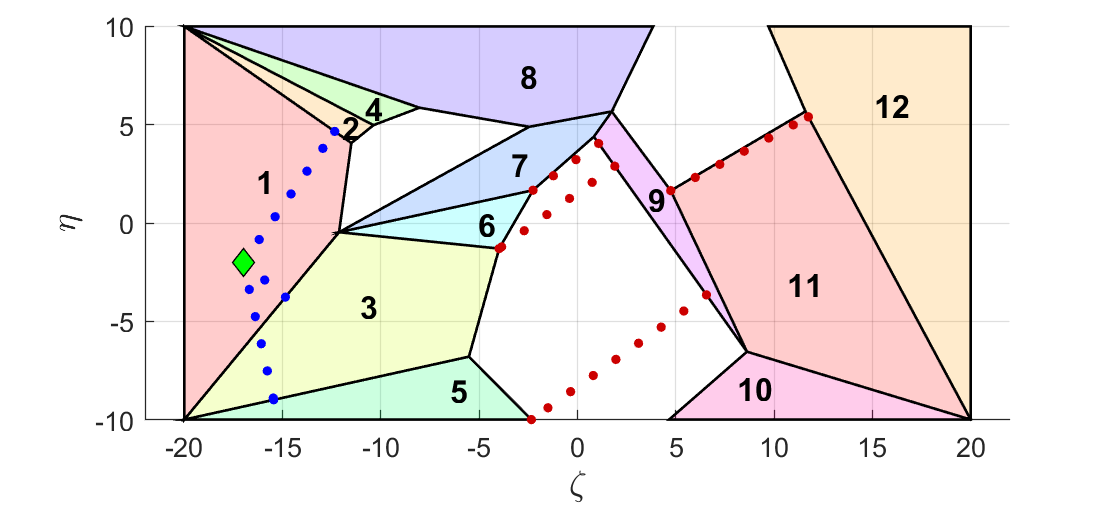}
    \caption{Obstacle avoidance map with convex regions displayed around white obstacles. The current AV position is displayed using a green diamond and the blue dots show point-to-region reachability for regions 2, 3, and 5. The red dots show region-to-region reachability for regions 3, 5, 6, and 12 from region 9. Referencing Definition~\ref{def:reachability}, the red and blue dots are spaced at an interval of $d_{max}$ for time steps less than $k_n$. 
    Given an assumed direction of travel, the plotted spacing between the final two dots is reduced as needed to prevent the final dot from jumping a region or entering its interior.}
    \label{fig:constraint_region_reachability}
\end{figure}

\begin{algorithm}[t]
\smallskip
Input: input ``region" or ``point" \\
\hphantom{Input:} number of time steps $k_n$ \\
\hphantom{Input:} max distance AV can travel in a time step $d_{max}$ \\
Result: reachable region indices $\mathbf{r}$
\begin{algorithmic}[1]
\Function{reachable}{input, $k_n$, $d_{max}$}
\State $\mathbf{r} \gets []$ 
\For{$r_i$ \textbf{in} $\{1, ..., n_F\}$}
    \Switch{input}
    \Case{region}
         \State \parbox[t]{\linewidth}{%
         $\mathbf{y}_0^*, \mathbf{y}_i^* \gets \argmin_{\mathbf{y}_0, \, \mathbf{y}_i} \left(||\mathbf{y}_0-\mathbf{y}_i||_2^2 \, \right|$ \\
         \hphantom{$\mathbf{y}_0^*, \mathbf{y}_i^* \gets$} $\left. \mathbf{y}_0 \in r, \, \mathbf{y}_i \in r_i \right)$ 
        } \label{algline:qp-region}
    \EndCase
    \Case{point}
        \State $\mathbf{y}_0^* \gets H \mathbf{x}_0$
         \State $\mathbf{y}_i^* \gets \argmin_{\mathbf{y}_i} \left(||\mathbf{y}_0^*-\mathbf{y}_i||_2^2 \, \middle| \mathbf{y}_i \in r_i \right)$ \label{algline:qp-point}
    \EndCase
    \EndSwitch
    \State distance $\gets ||\mathbf{y}_0^*-\mathbf{y}_i^*||_2$
    \If{distance $\leq d_{max} \cdot k_n$}
        \State \textbf{append} $r_i$ \textbf{to} $\mathbf{r}$
    \EndIf
\EndFor
\State \textbf{return} $\mathbf{r}$
\EndFunction
\end{algorithmic}
    \caption{Function to compute reachable regions}
    \label{alg:r2r-p2r-reachability}
\end{algorithm}

\subsubsection{Mixed-Integer Solver Implementation}
The structure of the mixed integer solver is provided in Algorithm~\ref{alg:mi-solver-structure}. For clarity of exposition, this algorithm only shows the case where the solver is configured to use a single thread (i.e., no parallelization). The solver can optionally be configured to use up to a maximum number of threads. In that case, each thread is responsible for solving a node and running the subsequent branch-and-bound logic. Additional solver settings include the absolute and relative convergence tolerances $\epsilon_a$ and $\epsilon_r$, the maximum solution time $t_{max}$, and the maximum number of branch-and-bound iterations $\text{it}_{\text{max}}$. In Algorithm~\ref{alg:mi-solver-structure}, $j_+$ denotes the objective of the incumbent and $j_-$ denotes the objective lower bound. Note that for a given objective function, larger values of $j_-$ will speed up convergence as discussed in Sec.~\ref{sec:conv-relax}.

\begin{algorithm}[t]
    Input: multi-stage MIQP as defined in \eqref{eq:miqp-multistage} \\
    \hphantom{Input:} reachable region lookup tables $R_p(k_n)$ and $R_r(k_n, r)$ \\
    \hphantom{Input:} solver settings $\epsilon_a$, $\epsilon_r$, $\text{it}_{max}$, $t_{max}$ \\
    Result: solution $\mathbf{z}_+$ with objective $j_+$ to MIQP
    \begin{algorithmic}[1]
        \State start timer $t$
        \For{$k \in \{0, ..., N\}$}
            \State root.$\mathbf{r}_k \gets R_p(k)$
        \EndFor
        \State \Call{MakeConsistent}{root}
        \State \Call{ApplyBinvars}{root}
        \State nodes, data.$j_+$ $\gets$ root, $+\infty$
        \State $\text{data}.R_r(k_n, r) \gets R_r(k_n, r)$
        \If{warm start enabled}
            \State \Call{ApplyBinvars}{warm start nodes}
            \State nodes $\gets$ \textbf{push} warm start nodes
        \EndIf
        \State it $\gets 0$
        \Do
            \State node $\gets$ \textbf{pop top}(nodes)
            \State solve(node)
            \State \Call{BranchAndBound}{node}
            \State it $\gets$ it $+ 1$
        \doWhile{\textbf{not} \Call{Converged}{} \textbf{and} nodes $\neq []$ \textbf{and} \Statex \qquad \qquad \, $t < t_{max}$ \textbf{and} it $< \text{it}_{\text{max}}$}
        \State \Return data.$\mathbf{z}_+$, data.$j_+$
        \Function{Converged}{}
            \State $j_- \gets$ min (node.$j$ for node \textbf{in} nodes)
            \State $j_+ \gets \text{data}.j_+$
            \State \Return ($|j_{+}-j_{-}|/|j_{+}| < \epsilon_r$ \textbf{or} $|j_{+}-j_{-}| < \epsilon_a$)
        \EndFunction
    \end{algorithmic}
    \caption{MIQP solver structure}
    \label{alg:mi-solver-structure}
\end{algorithm}

This solver implements a modified best bound node selection strategy using a sorted data structure. Nodes on top of the data structure have priority for evaluation, and a \textbf{push-top} operation is used to give priority to a node independent of its cost function lower bound.
This data structure is denoted as \textbf{nodes}. Similarly, the problem data is denoted as \textbf{data}. Both \textbf{nodes} and \textbf{data} are available to all algorithms and threads.

Warm starting is implemented based on the optimal region vector $\bar{\mathbf{r}}$ from the previous solution and the reachable region lookup table $R_r(k_n, r)$. The warm-started region selections are $\bar{\mathbf{r}}_i^{WS} = \begin{bmatrix} r_1 & \cdots & r_N & r_{i} \end{bmatrix}$ where $r_i \in R_r(1, r_N)$. One warm start node is generated for each $\bar{\mathbf{r}}_i^{WS}$. 

Algorithm~\ref{alg:core_bnb} describes the specific branch-and-bound logic used within this solver. Here, the depth of a node is defined to be the number of time steps of the MPC problem for which the associated binary variables are integer feasible. Equivalently, this is the number of time steps for which the number the number of elements of $\mathbf{r}_k$ is equal to one. 

This branch-and-bound algorithm implements branching heuristics based on the nodal solution $\mathbf{z}$. These heuristics make use of a \textproc{GetFirstConsViolation}(node) function. This function iterates over $k$ from $\{0, ..., N\}$ to check the $\mathbf{x}_k$ trajectory within the nodal solution against a halfspace representation of the constraints given in \eqref{eq:mpc-gen-obs-avoid} (i.e., one H-rep polytope for each convex region). The time step of the first detected constraint violation is given as $k^v$ and the region selections corresponding to $\mathbf{x}_k$ for $k < k^v$ are given as $\bar{\mathbf{r}}^{nv}$. If there are no constraint violations, then the cost of enforcing the region selections $\bar{\mathbf{r}}^{nv}$ is estimated. To do this, the node solution $\mathbf{z}_k$ is modified by applying the region selections $\bar{\mathbf{r}}_k^{nv}$ and the cost \eqref{eq:miqp-multistage-cost} is computed. If this estimated cost is lower than that of the incumbent, then a node is created with these region selections and is given priority for evaluation. The function to estimate the cost of enforcing $\bar{\mathbf{r}}^{nv}$ is denoted as \textproc{ApproxCost}($\mathbf{z}$, $\bar{\mathbf{r}}$).

\begin{algorithm}[t]
Input: solved node \\
Result: updated nodes data structure and incumbent solution
\begin{algorithmic}[1]
\Procedure{BranchAndBound}{node}
    \If{node \textbf{is not} feasible \textbf{or} $\text{node}.j > \text{data}.j_+$}
        \State \Return
    \ElsIf{ (depth(node) $= N+1$ \textbf{and}
        \Statex \qquad \qquad \quad $\text{node}.j < \text{data}.j_+$)
    }
        \State data.$\mathbf{z}_+, \text{data}.j_+ \gets \text{node}.\mathbf{z}, \text{node}.j$ 
        \State \textbf{prune} nodes($\text{node}.j > \text{data}.j_+$) \textbf{from} nodes
        \State \Return
    \Else
        \State $\text{cv}, \bar{\mathbf{r}}^{nv}, k^v \gets$ \Call{GetFirstConsViolation}{node}
        \If{\textbf{not} cv}
            \State $j_a \gets$ \Call{ApproxCost}{$\text{node}.\mathbf{z}, \bar{\mathbf{r}}^{nv}$}
            \If{$j_a < \text{data}.j_+$}
                \State node.$\bar{\mathbf{r}}\gets \bar{\mathbf{r}}^{nv}$
                \State nodes $\gets$ \textbf{push top}(node)
            \EndIf
        \EndIf
        \If{cv}
            \State \Call{BranchAtK}{node, $k^v$}
        \Else
            \State \Call{BranchMF}{node}
        \EndIf
    \EndIf
\EndProcedure
\end{algorithmic}
    \caption{Core branch and bound logic}
    \label{alg:core_bnb}
\end{algorithm}

Different branching strategies are used depending on whether \textproc{GetFirstConsViolation}(node) detects a constraint violation. For the case that a constraint violation is detected, Algorithm~\ref{alg:branch_at_timestep} is used. This branching strategy is well suited to obstacle avoidance problems as discussed in \cite{robbins2024efficient} and is similar to an algorithm presented in \cite{fukushima2013model}. At the specified time step, Algorithm~\ref{alg:branch_at_timestep} branches based on the binary variable with the highest fractional value in the relaxed nodal solution. For the case where no constraint violation is detected, Algorithm~\ref{alg:branch_most_frac} is used. This is a variant of the widely used most-fractional branching rule \cite{karamanov2006branch, stellato2018embedded, breu2009branch}. Algorithm~\ref{alg:branch_most_frac} selects a branching time step based on how far the fractional solution is from being integer-valued at that time step. This more general branching logic is used to account for cases where the obstacle avoidance-based branching logic (Algorithm~\ref{alg:branch_at_timestep}) does not apply, such as when using region-dependent costs.

\begin{algorithm}[t]
Input: solved node \\
\hphantom{Input:} branching time step $k^v$ \\
Result: updated nodes data structure
\begin{algorithmic}[1]
    \Procedure{BranchAtK}{node, $k^v$}
        \State $k \gets k^v$
        \State $\text{node}_b, \text{node}_r \gets \text{node}$
        \State $r_{\text{max}} \gets \argmax_r \left(\text{node}.\bm{\xi}_{b k}[r] \right)$ \label{algline:branch_at_timestep_argmax_b}
        \State $\text{node}_b.\mathbf{r}_{k} \gets r_{\text{max}}$
        \State $\text{node}_r.\mathbf{r}_{k} \gets \text{node}.\mathbf{r}_k \setminus r_{\text{max}}$
        \State \Call{MakeConsistent}{$\{\text{node}_b, \text{node}_r\}$}
        \State \Call{ApplyBinvars}{$\{\text{node}_b, \text{node}_r\}$}
        \State nodes $\gets$ \textbf{push} $\{ \text{node}_b, \text{node}_r \}$ 
    \EndProcedure
\end{algorithmic}
    \caption{Branch based on fractional solution at specified time step}
    \label{alg:branch_at_timestep}
\end{algorithm}

\begin{algorithm}[t]
Input: solved node \\
Result: updated nodes data structure
\begin{algorithmic}[1]
    \Procedure{BranchMF}{node}
        \State $\bm{\xi}_{bk}^f \gets |\text{node}.\bm{\xi}_{bk} - \text{round}(\text{node}.\bm{\xi}_{bk})|$
        \State $k^b \gets \argmax_k \left[ \max_r \left( \bm{\xi}_{bk}^f[r] \right) \right]$ 
        \State \Call{BranchAtK}{node, $k^b$}
    \EndProcedure
\end{algorithmic}
    \caption{Branch at timestep where solution is most fractional}
    \label{alg:branch_most_frac}
\end{algorithm}

Algorithm~\ref{alg:make_reach_consistent} enforces consistency of the reachability constraints to set additional binary variables to zero where possible. This algorithm leverages the property that the QP relaxation of a node is its convex hull (see Sec.~\ref{sec:conv-relax}) to generate tighter sub-problems, reducing the number of branch-and-bound iterations required to converge. 
Operators $\&$ and $|$ denote the element-wise \textbf{and} and \textbf{or} operations, respectively. The \textproc{List2Array}($\mathbf{r}$) and \textproc{Array2List}($\mathbf{a}$) functions convert between a vector of region indices $\mathbf{r}$ and a fixed-length array of boolean variables $\mathbf{a}$ for which $a_i = \text{True}$ if $i \in \mathbf{r}$ and $a_i = \text{False}$ otherwise.

\begin{algorithm}[t]
Input: unsolved node \\
Result: node with updated region indices $\bar{\mathbf{r}}$
\begin{algorithmic}[1]
    \Procedure{MakeConsistent}{node}
        \For{$k$ \textbf{in} $\{0, ..., N\}$}
            \State $\mathbf{a}_k \gets$ \Call{list2array}{$\mathbf{r}_k$}
            \For{$k_o$ \textbf{in} $\{0, ..., N\}$}
                \State $\mathbf{a}_r \gets \mathbf{0}^{1 \times n_b}$
                \For{$r_{ko,i}$ \textbf{in} $\mathbf{r_{ko}}$}
                    \State $\mathbf{r}_r \gets \text{data}.R_r(|k-k_o|, r_{ko,i})$
                    \State $\mathbf{a}_r \gets \mathbf{a}_r$ $|$ \Call{list2array}{$\mathbf{r}_r$}
                \EndFor
                \State $\mathbf{a}_k \gets \mathbf{a}_k$ $\&$ $\mathbf{a}_r$
            \EndFor        
            \State $\mathbf{r}_k \gets$ \Call{array2list}{$\mathbf{a}_k$}
        \EndFor
    \EndProcedure
\end{algorithmic}
    \caption{Make reachability constraints consistent}
    \label{alg:make_reach_consistent}
\end{algorithm}

Given a node with region indices $\bar{\mathbf{r}}$, the corresponding QP sub-problem is generated by updating the multi-stage MIQP problem definition in \eqref{eq:miqp-multistage}. The function which does this is denoted as \textproc{ApplyBinvars}(node) in Algorithm~\ref{alg:mi-solver-structure}. In this function, variables that must be zero given $\bar{\mathbf{r}}$ are identified. Zeroed binary variables $\bm{\xi}_{bk}$ are readily identified from $\bar{\mathbf{r}}$ given its definition. The mixed-integer solver additionally checks for constraints of the form 
\begin{equation} \label{eq:remove-cont-vars}
\sum_{i \in \mathcal{I}} \gamma_i z_i = 0 \;,\; \gamma_i > 0 \;,\; z_i \geq 0 \;.
\end{equation}
Eq.~\ref{eq:remove-cont-vars} implies $z_i = 0 \; \forall i \in \mathcal{I}$. Constraints of this form arise when binary factors are eliminated from hybrid zonotopes generated using \cite[thm 5]{siefert2023reachability}. Columns and rows of $P_k$, $\mathbf{q}_k$, $C_k$, $D_k$, $\mathbf{c}_k$, $G_k$, and $\mathbf{w}_k$ corresponding to the zeroed variables are removed. Any constraints that depend only on the zeroed variables are also removed. The remaining binary variables are then relaxed to continuous variables to generate the QP sub-problem. In combination with the reachability constraints, this method of generating QP sub-problems improves the scalability of the MIQP solver with respect to map size by removing variables corresponding to portions of the map that cannot be reached. 

QP sub-problems generated in this way have output constraints of the form $\mathbf{y}_k \in \mathcal{F}^{CR}_k$ where $\mathcal{F}^{CR}_k = \mathit{CH}(\bigcup_{r \in \mathbf{r}_k} \mathcal{F}_r)$. In other words, the QP sub-problems are tight with respect to the selected set of convex regions $\mathbf{r}_k$. To see this, consider an example with two convex regions given in V-rep as $\mathcal{F}_1 = \{\mathbf{v}_1, \mathbf{v}_2\}$ and $\mathcal{F}_2 = \{\mathbf{v}_2, \mathbf{v}_3\}$. From \eqref{eq:vrep2hybzono_sets}, the hybrid zonotope representing this union is
\begin{multline} \label{eq:hybzono-2-line-segments}
    \mathcal{F} = \left\langle \begin{bmatrix}
        \mathbf{v}_1 & \mathbf{v}_2 & \mathbf{v}_3 & \mathbf{0} & \mathbf{0} & \mathbf{0}
    \end{bmatrix}, \begin{bmatrix}
        \mathbf{0} & \mathbf{0}
    \end{bmatrix}, \mathbf{0}, \right. \\ 
    \left.
    \begin{bmatrix}
        1 & 1 & 1 & 0 & 0 & 0 \\
        0 & 0 & 0 & 0 & 0 & 0 \\
        1 & 0 & 0 & 1 & 0 & 0 \\
        0 & 1 & 0 & 0 & 2 & 0 \\
        0 & 0 & 1 & 0 & 0 & 1
    \end{bmatrix}, \begin{bmatrix}
        0 & 0 \\
        1 & 1 \\
        -1 & 0 \\
        -1 & -1 \\
        0 & -1
    \end{bmatrix}, \begin{bmatrix}
        1 \\ 1 \\ 0 \\ 0 \\ 0
    \end{bmatrix}  \right\rangle \;.
\end{multline}
Referencing \eqref{eq:C_k-hybzono}, the hybrid zonotope equality constraints $A_c^{hz} \bm{\xi}_{ck} + A_b^{hz} \bm{\xi}_{bk} = \mathbf{b}^{hz}$ appear directly in the multi-stage MIQP formulation with $\bm{\xi}_{ck}$ and $\bm{\xi}_{bk}$ as optimization variables. If $\mathbf{r}_k = 1$, this implies that $\xi_{bk,2}=0$. Applying \eqref{eq:remove-cont-vars} to \eqref{eq:C_k-hybzono} at time step $k$ identifies $\xi_{ck,3}$ and $\xi_{ck,6}$ as continuous variables that must be zero. Updating the MIQP definition in \eqref{eq:miqp-multistage} to account for the zeroed variables as described above results in $C_k$, $\mathbf{c}_k$, $G_k$, and $\mathbf{w}_k$ that encode the constraint $\mathbf{y} \in \mathcal{F}_{\mathbf{r} k}$ with 
\begin{multline} \label{eq:hybzono-2-line-segments-reduced}
    \mathcal{F}_{\mathbf{r} k} = \left\langle \begin{bmatrix}
        \mathbf{v}_1 & \mathbf{v}_2 & \mathbf{0} & \mathbf{0}
    \end{bmatrix}, \begin{bmatrix}
        \mathbf{0}
    \end{bmatrix}, \mathbf{0}, \right. \\ 
    \left.
    \begin{bmatrix}
        1 & 1 & 0 & 0 \\
        0 & 0 & 0 & 0 \\
        1 & 0 & 1 & 0 \\
        0 & 1 & 0 & 2 \\
    \end{bmatrix}, \begin{bmatrix}
        0 \\
        1 \\
        -1 \\
        -1 \\
    \end{bmatrix}, \begin{bmatrix}
        1 \\ 1 \\ 0 \\ 0
    \end{bmatrix}  \right\rangle \;.
\end{multline}
The hybrid zonotope $\mathcal{F}_{\mathbf{r} k}$ can be seen to be of the form given in \eqref{eq:Z-vrep-gen} with vertex and incidence matrices $V$ and $M$ defined for the convex regions $\mathcal{F}_r$,  $r \in \mathbf{r}_k$. As such, Theorem~\ref{thm:conv-relax-vrep} applies and $\mathit{CR}(\mathcal{F}_{\mathbf{r} k}) = \mathit{CH}(\mathcal{F}_{\mathbf{r} k})$ with $\mathcal{F}_{\mathbf{r} k} = \bigcup_{r \in \mathbf{r}_k} \mathcal{F}_r$.

The same result can be shown to hold for hybrid zonotopes used to represent OGMs. For example, consider an OGM with three cells defined by the hybrid zonotope 
\begin{equation}
    \mathcal{F} = \left\langle \begin{bmatrix}
        d_x & 0 \\ 0 & d_y
    \end{bmatrix}, \begin{bmatrix}
        \mathbf{c}_{b1} & \mathbf{c}_{b2} & \mathbf{c}_{b3}
    \end{bmatrix}, \begin{bmatrix}
        -d_x/2 \\ -d_y/2
    \end{bmatrix}, 
    \mathbf{0}^T, \mathbf{1}^T, 1 \right\rangle \;.
\end{equation}
If $\mathbf{r}_k = \{1, 2\}$, then $\xi_{bk,3} = 0$ and the resulting hybrid zonotope defined within $C_k$, $\mathbf{c}_k$, $G_k$, and $\mathbf{w}_k$ will be
\begin{equation}
    \mathcal{F}_{\mathbf{r} k} = \left\langle \begin{bmatrix}
        d_x & 0 \\ 0 & d_y
    \end{bmatrix}, \begin{bmatrix}
        \mathbf{c}_{b1} & \mathbf{c}_{b2}
    \end{bmatrix}, \begin{bmatrix}
        -d_x/2 \\ -d_y/2
    \end{bmatrix}, 
    \mathbf{0}^T, \mathbf{1}^T, 1 \right\rangle \;.
\end{equation}
Theorem~\ref{thm:conv-relax-occ-grid} applies, so $\mathit{CR}(\mathcal{F}_{\mathbf{r} k}) = \mathit{CH}(\mathcal{F}_{\mathbf{r} k})$ with $\mathcal{F}_{\mathbf{r} k} = \bigcup_{r \in \mathbf{r}_k} \mathcal{F}_r$.

\subsection{Quadratic Program Solver} \label{sec:qp_solver}

A custom implementation of Mehrotra's primal-dual interior point method \cite{mehrotra1992implementation, borrelli2017predictive} is used to solve the QP sub-problems generated by the mixed-integer solver (Sec.~\ref{sec:mi_solver}).

As described in \cite{wang2009fast} and \cite{domahidi2012efficient}, MPC QPs have a multi-stage structure that can be exploited in interior point solvers. Interior point methods solve a large linear system to compute the Newton step at each iteration. This is typically the costliest operation in the solver. Multi-stage algorithms break this linear system into a series of smaller linear systems that can be solved much more efficiently. The proposed solver implements the Newton step solution algorithm provided in \cite{domahidi2012efficient}. As pointed out in \cite{robbins2024efficient}, the Newton step solution strategy provided in \cite{domahidi2012efficient} can be accelerated by using constrained zonotopes to represent any polytopic constraints. 

\subsubsection{Efficient Solution of Multi-Stage QPs}
Here, the procedure to solve for the Newton step, following \cite{domahidi2012efficient}, is described. For conciseness, Mehrotra's algorithm is not presented here, and the interested reader is directed to \cite{borrelli2017predictive}.

Consider a QP relaxation of \eqref{eq:miqp-multistage} with $\bm{\xi}_{bk}$ relaxed from $\integer^{n_b} \rightarrow \real^{n_b}$. In primal-dual interior point methods, the Newton step is given by the solution to the linear system
\begin{equation} 
\begin{bmatrix}
    P & C^T & G^T & 0 \\
    C & 0 & 0 & 0 \\
    G & 0 & 0 & I \\
    0 & 0 & S & \Lambda
\end{bmatrix}
\begin{bmatrix} \Delta \mathbf{z} \\ \Delta \bm{\nu} \\ \Delta \bm{\lambda} \\ \Delta \mathbf{s} \end{bmatrix} = -
\begin{bmatrix} \mathbf{r}_C \\ \mathbf{r}_E \\ \mathbf{r}_I \\ \mathbf{r}_S \end{bmatrix} \;,\label{eq:Newton_step_sparse}
\end{equation} where
\begin{subequations} \label{eq:qp-matrix-defs}
\begin{align} 
P &= \text{blkdiag} \left( \begin{bmatrix} P_0 & \cdots & P_{N} \end{bmatrix} \right) \, , \label{eq:hessian_matrix} \\   
C &= \begin{bmatrix}
C_0 & D_1 & \cdots & 0 \\
0 & \ddots & \ddots & 0 \\
0 & \cdots & C_{N-1} & D_{N}
\end{bmatrix} \, , \label{eq:eq_cons_matrix} \\
G &= \text{blkdiag} \left( \begin{bmatrix} G_0 & \cdots & G_{N} \end{bmatrix} \right) \, , \label{eq:ineq_cons_matrix} \\   
\mathbf{r}_C &= P \mathbf{z} + \mathbf{q} + C^T \bm{\nu} + G^T \bm{\lambda} \, , \\
\mathbf{r}_E &= C \mathbf{z} + \mathbf{c} \, , \\
\mathbf{r}_I &= (G \mathbf{z} - \mathbf{w}) + \mathbf{s} \, , \\
\mathbf{r}_S &= S \bm{\lambda} \, .
\end{align}
\end{subequations}
Here, $\mathbf{z}$ are the primal variables, $\bm{\nu}$ and $\bm{\lambda}$ are the dual variables for the equality  and inequality constraints, respectively, and $\mathbf{s}$ are the inequality constraint slack variables. $S$ and $\Lambda$ are diagonal matrices constructed from $\mathbf{s}$ and $\bm{\lambda}$.

Solving \eqref{eq:Newton_step_sparse} directly generally requires an LU decomposition, which has a cost of $(2/3) n^3$ flops where $n$ is the dimension of the linear system matrix. Although sparse linear algebra techniques can be used to reduce the computational cost, solving \eqref{eq:Newton_step_sparse} remains the primary computational challenge in primal-dual interior point methods.

For multi-stage problems, \eqref{eq:Newton_step_sparse} can be written concisely as
\begin{equation} \label{eq:Newton_step_concise}
    Y \Delta \bm{\nu} = \bm{\beta} \:,
\end{equation}
with
\begin{subequations} \label{eq:Y_beta}
\begin{align}
&Y = C \Phi^{-1} C^T \in \mathcal{S}_{>0} \;, \\
&\bm{\beta} = \mathbf{r}_E - C \Phi^{-1} \left( \mathbf{r}_C + G^T S^{-1} \Lambda \mathbf{r}_I - G^T S^{-1} \mathbf{r}_S \right) \;,
\end{align}
\end{subequations}
and $\mathcal{S}_{>0}$ denoting the set of positive definite matrices. $Y$ has block-banded structure given by
\begin{subequations} \label{eq:Y_kk}
\begin{align}
Y_{kk} =& C_{k-1} \Phi_{k-1}^{-1} C_{k-1}^T + D_{k} \Phi_{k}^{-1} D_{k}^T \nonumber \\
& \forall k \in \{1, \cdots, N\} \;, \\
Y_{k, k+1} =& D_k \Phi_k^{-1} C_k^T \; \forall k \in \{1, \cdots, N-1\} \;, \\
Y_{k+1, k} =& Y_{k, k+1}^T,\; Y_{kj} = 0 \; \forall j \notin \{k-1, k, k+1\} \;.
\end{align}
\end{subequations}
The matrix $\Phi = \text{blkdiag} \left( \begin{bmatrix} \Phi_0 & \cdots & \Phi_{N} \end{bmatrix} \right)$ has blocks
\begin{equation} \label{eq:Phi_k}
\Phi_k = P_k + G_k^T S_k^{-1} \Lambda_k G_k \;.
\end{equation}
Matrices $S_k$ and $\Lambda_k$ correspond to the inequality constraints at time step $k$.

Matrix $\Phi_k$ must be invertible to evaluate \eqref{eq:Y_beta} and \eqref{eq:Y_kk}. The $P_k$ matrix is singular (though positive semi-definite) as defined in \eqref{eq:cost_fcn_Pk}, so invertibility must be ensured by the $G_k^T S_k^{-1} \Lambda_k G_k$ term. Proposition~\ref{prop:Phi_k_pos_def} shows that this is achieved by requiring $G_k$ to be full column rank, which is indeed the case based on its definition in \eqref{eq:ineq_cons_Gk}.

\begin{proposition} \label{prop:Phi_k_pos_def}
If $G_k$ is full column rank and $P_k$ is positive semi-definite, then $\Phi_k$ is positive definite and therefore invertible.    
\begin{proof}
$S_k^{-1} \Lambda_k$ is positive definite because it is a diagonal matrix with all elements $(S_k^{-1} \Lambda_k)_{ii} > 0$. If $G_k$ is full column rank, then $G_k^T S_k^{-1} \Lambda_k G_k$ is positive definite using the property that a $r \times r$ matrix $B A B^T$ is positive definite for a $n \times n$ matrix $A \in \mathcal{S}_{>0}$ and a $r \times n$ matrix $B$ with $\text{rank}(B)=r$ \cite{petersen2008matrix}. $\Phi_k$ is positive definite because the sum of a positive semi-definite matrix and a positive definite matrix is positive definite. 
\end{proof}
\end{proposition}

To construct the $Y_{kj}$ matrices, it is necessary to first compute Cholesky factorizations of the $N+1$ $\Phi_k$ matrices. The $Y_{kj}$ can then be constructed by substitution. Each $\Phi_k$ factorization has a cost of $(1/3) n_{in\,k}^3$ flops where $n_{in\,k}$ is the number of inequality constraints for the $k^{th}$ MPC time step.

Once the $Y$ matrix has been constructed, it is factorized to solve \eqref{eq:Newton_step_concise}. This requires $N$ Cholesky factorizations at a cost of $(1/3) n_{eq\,k}^3$ flops where $n_{eq\,k}$ is the number of equality constraints in the $C_k$ matrix.

\subsubsection{Efficient QP Solution with Constrained Zonotope Constraints} \label{sec:qp-cons-zono}
When using hybrid zonotopes to represent the obstacle-free space, QP relaxations of \eqref{eq:miqp-multistage} have constrained zonotope constraints as defined in \eqref{eq:cons_zonotope}.

Constraints of the form $\mathbf{y}_{k} \in \mathcal{Z}_C$ with $\mathcal{Z_C}$ a constrained zonotope can be exploited when computing the Cholesky factorizations of the $\Phi_k$ matrices during the Newton step computation. As stated in \cite{domahidi2012efficient}, the $\Phi_k$ matrices are diagonal for the case of diagonal quadratic cost $P_k$ and box constraints $G_k$.
$P_k$ is diagonal by Assumption~\ref{as:Pk-diag}. $G_k$ is in box constraint form in our problem formulation as discussed in Sec.~\ref{sec:baseline_prob_form} and shown in \eqref{eq:ineq_cons_Gk}. In the QP solver implemented here, a low-cost diagonal matrix inversion of $\Phi_k$ is computed when needed rather than computing the Cholesky factors of $\Phi_k$.

The computational penalty in the QP solver for using constrained zonotopes is the inclusion of additional equality constraints and primal optimization variables in \eqref{eq:miqp-multistage} as described in Sec.~\ref{sec:miqp}. The equality constraints will increase the computational cost of factorizing the $Y_{kk}$ matrices. The $\bm{\xi}_{ck}$ and relaxed $\bm{\xi}_{bk}$ will increase the overall size of the linear system in \eqref{eq:Newton_step_sparse}. However, in the multi-stage formulation, the operations with the worst scaling (Cholesky factorizations at $(1/3)n^3$ flops) are in terms of $n_{in\,k}$ and $n_{eq\,k}$. As such, the addition of the $\bm{\xi}_{ck}$ and  $\bm{\xi}_{bk}$ variables does not significantly affect the QP solution time for large problem sizes. 
\section{Results} \label{sec:results}

\subsection{Desktop Computer Simulation} \label{sec:results-sim}
This section gives simulation results for a discrete-time double integrator dynamics model with an MPC controller. Double integrator models are frequently used for AV motion planning \cite{matni2024quantitative, quirynen2023real, whitaker2021optimal, gratzer2024two, agrawal2021constructive}. Two MIQP solvers were used: The solver proposed in this article and the state-of-the-art commercial solver Gurobi \cite{gurobi}. Furthermore, two MPC formulations were used: One using a hybrid zonotope to represent the obstacle-free space $\mathcal{F}$, and one using a union of H-rep polytopes. Hyperplane arrangements were not considered for a free space partition given their similarity in structure to unions of H-rep polytopes. In \cite{bird2021unions}, hybrid zonotope obstacle-free space descriptions were found to result in reduced MPC solution times when compared to hyperplane arrangements when both formulations were solved with Gurobi. 

The proposed solver was implemented in C++ and depends on the C++ standard library and the Eigen linear algebra library \cite{eigenweb}. Gurobi was called via its C++ API. Inheritance was used to ensure that the same MIQPs were passed to both the proposed solver and Gurobi (i.e., the same code was used to build \eqref{eq:miqp-multistage} in either case). H-rep formulations were constructed as described in \cite{robbins2024efficient}.

An MPC prediction horizon of $N=15$ steps was used. Both the proposed MIQP solver and Gurobi were configured to use relative and absolute convergence tolerances of $\epsilon_a = 0.1$ and $\epsilon_r = 0.01$ respectively (see Algorithm~\ref{alg:mi-solver-structure}). 

\begin{figure*}[t!]
\centering
\begin{minipage}{0.6\textwidth}
\includegraphics[width=\textwidth]{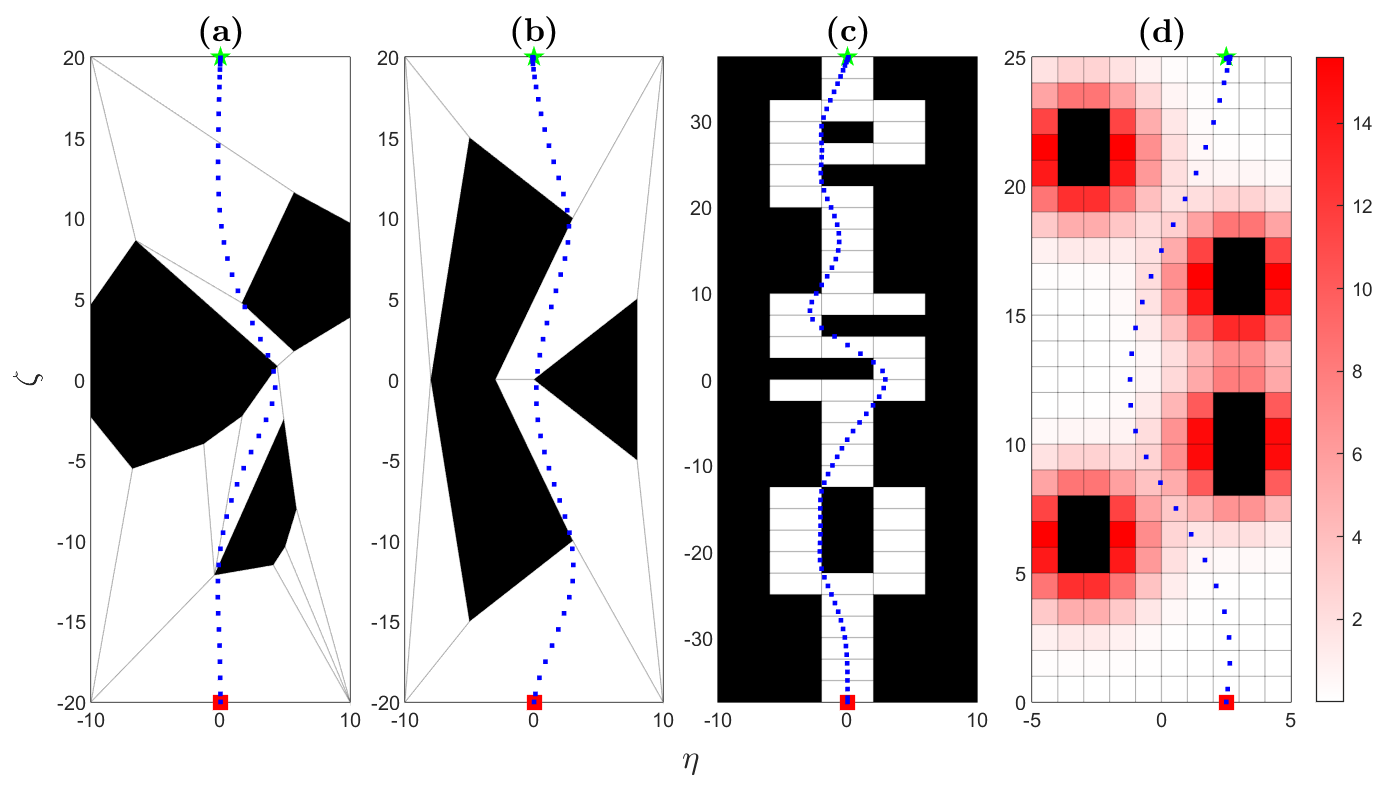}
\end{minipage}
\begin{minipage}{0.3\textwidth}
\includegraphics[width=\textwidth]{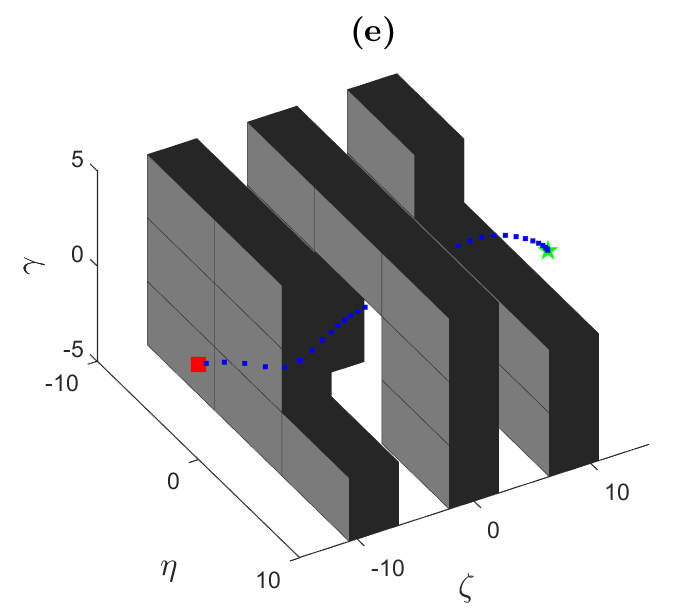}
\end{minipage}
\caption{Simulation results for double integrator model using: \textbf{(a)} a polytopic map with convex obstacles, \textbf{(b)} a polytopic map with non-convex obstacles, \textbf{(c)} a binary OGM, \textbf{(d)} an OGM with both obstacles and cell-dependent costs, and \textbf{(e)} a binary OGM with three spatial dimensions.} 
\label{fig:trajectories}
\end{figure*}

Two-dimensional and three-dimensional discrete-time double integrator dynamics models were considered. For the two-dimensional case, the linear dynamics matrices are given by
\begin{equation} \label{eq:dbl_int_A_B}
A = \begin{bmatrix}
    1 & \Delta t & 0 & 0 \\
    0 & 1 & 0 & 0 \\
    0 & 0 & 1 & \Delta t \\
    0 & 0 & 0 & 1
\end{bmatrix},\; 
B = \begin{bmatrix} \frac{1}{2} \Delta t^2 & 0 \\ \Delta t & 0 \\ 0 & \frac{1}{2} \Delta t^2 \\ 0 & \Delta t\end{bmatrix}\;.
\end{equation}
A discrete time step of $\Delta t=1$ was used. The state and input are $\mathbf{x}_k = \begin{bmatrix}
    \zeta_k & \dot{\zeta}_k & \eta_k & \dot{\eta}_k
\end{bmatrix}^T$ and $\mathbf{u}_k = \begin{bmatrix}
    \ddot{\zeta}_k & \ddot{\eta}_k
\end{bmatrix}^T$, respectively, where $\zeta$ and $\eta$ are position coordinates. The state, input, and terminal state constraint sets are defined using constrained zonotopes as described in Sec.~\ref{sec:problem-formulation}. Referencing \eqref{eq:state-input-cons-conzono}, the state and input constraints for the two-dimensional case are given by
\begin{subequations} \label{eq:dbl_int_conzono}
\begin{align}
&\mathcal{Z}_{cx} = \langle v_{\text{max}} I_2, \mathbf{0}_2, [], [] \rangle \;, H_x = \begin{bmatrix}
    0 & 1 & 0 & 0 \\
    0 & 0 & 0 & 1
\end{bmatrix} \;, \\
&\mathcal{Z}_{cu} = \langle a_{\text{max}} I_2, \mathbf{0}_2, [], [] \rangle \;, H_u = I_2 \;, \\
&\mathcal{Z}_{cxT} = \langle 0_{2\times2}, \mathbf{0}_2, [], [] \rangle \;, H_{xT} = \begin{bmatrix}
    0 & 1 & 0 & 0 \\
    0 & 0 & 0 & 1
\end{bmatrix} \;,
\end{align}
\end{subequations}
where $v_{max}=1$ and $a_{max}=1$ are velocity and acceleration limits, respectively. The terminal constraint requires the vehicle to be at rest. All box constraints in \eqref{eq:state-input-cons-conzono} were set to $\pm 1e4 \cdot \mathbf{1}$. Constraint softening was used for the state and terminal state constraint sets with $W_{cxk} = W_{cxN} = 1e6 \cdot I_2$ and $\bm{\sigma}_{cx k}^u = \bm{\sigma}_{cx N}^u = 1e4 \cdot \mathbf{1}$. The obstacle avoidance constraints were additionally subject to softening with $W_k = 1e6 \cdot I$ and $\bm{\sigma}^u = 1e4 \cdot \mathbf{1}$ with dimensions dependent on the problem and constraint formulation. The MPC cost function matrices are given by
\begin{subequations} \label{eq:dbl_int_Q_R}
\begin{align}
&Q_k = \text{diag}(\begin{bmatrix}
    0.1 & 0 & 0.1 & 0
\end{bmatrix})\;, \\
&R_k = 10 \cdot I_2\;, \\
&Q_N = \text{diag}(\begin{bmatrix}
    10 & 0 & 10 & 0
\end{bmatrix}) \;.
\end{align}
\end{subequations}

The three-dimensional double integrator dynamics model was created using straightforward extensions of \eqref{eq:dbl_int_A_B}, \eqref{eq:dbl_int_conzono}, and \eqref{eq:dbl_int_Q_R}.

Simulated trajectories for the double integrator model are given in Fig.~\ref{fig:trajectories}. These trajectories were generated using the proposed solver with hybrid zonotope constraints. The path taken by the model is given by blue dots. The starting position is the red square while the reference position is the green star. The black regions are obstacles. The free space is partitioned into convex sub-regions, whose boundaries are displayed as gray lines. The free space partition is not displayed for map \textbf{(e)} but can be deduced from the OGM formulation. In map \textbf{(d)}, the costs $\mathbf{q}_k^r$ of different region selections are color-coded using varying shades of red. These trajectories can be observed to pass through corners of obstacles because of the discrete-time formulation. Obstacles may be interpreted as having been bloated to account for vehicle geometry and inter-sample behavior as in \cite{guthrie2022closed} such that this corner-cutting behavior does not result in a collision.

\begin{figure*}[t!]
\centering
\includegraphics[width=\linewidth]{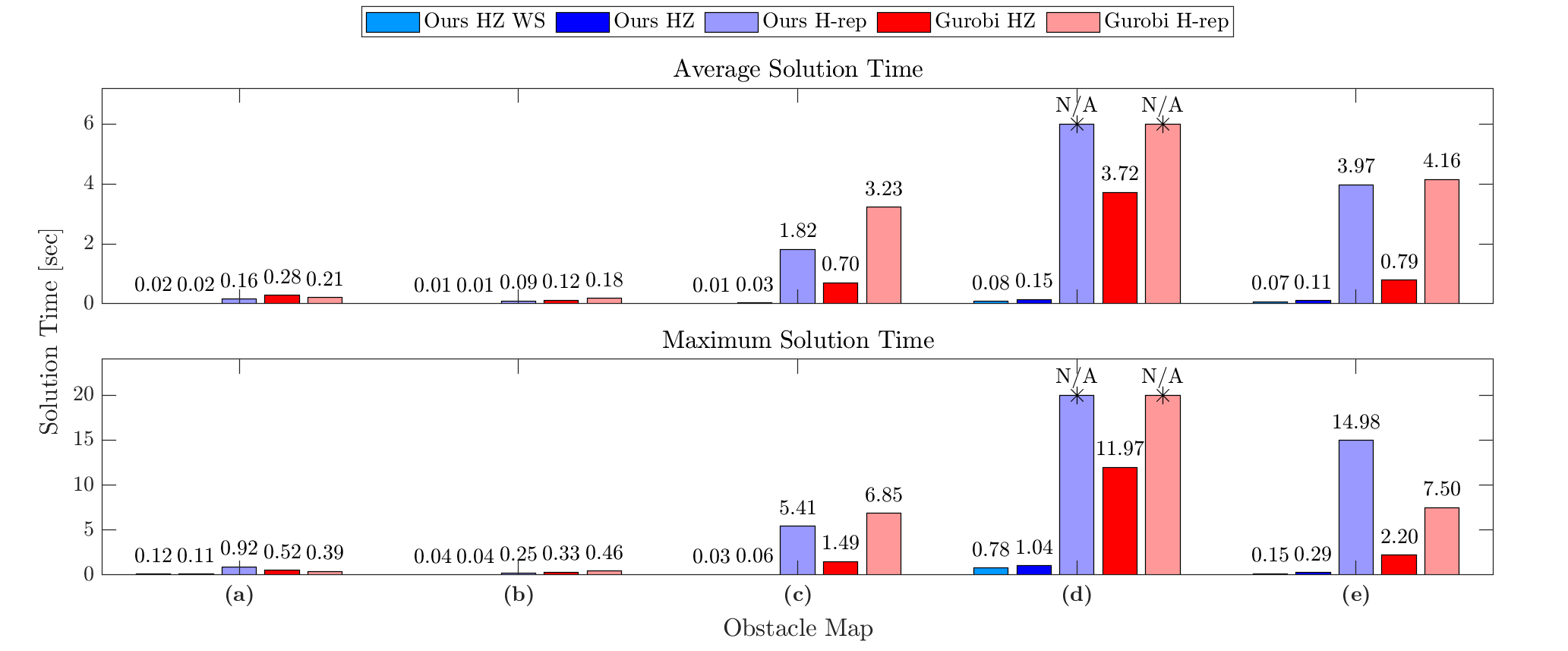}
\caption{Average and maximum solution times for the double integrator example using an MPC horizon of $N=15$. The obstacle maps correspond to those shown in Fig.~\ref{fig:trajectories}. ``HZ" denotes a hybrid zonotope representation of the obstacle free space, while ``H-rep" denotes a union of H-rep polytopes using the Big-M method. ``WS" indicates that the binary variables were warm started.}
\label{fig:sol-times}
\end{figure*}

Fig.~\ref{fig:sol-times} compares solution times of the MPC optimization problem across different solvers (the proposed MIQP solver vs. Gurobi \cite{gurobi}) and different obstacle-free space representations (hybrid zonotopes (HZ) vs. unions of H-rep polytopes using the Big-M method) for the simulated double integrator dynamics described above. Obstacle maps \textbf{(a)}-\textbf{(e)} correspond to those shown in Fig.~\ref{fig:trajectories}.

All results presented in this section (Sec.~\ref{sec:results-sim}) use the following parameters and simulation setup. Both solvers were given a maximum solution time of 60~sec. If they failed to converge within that time on any iteration, then the simulation was terminated, and that case was marked as ``N/A" in the figures and tables. An asterisk was placed on top of bars for cases where the value of the bar does not fit within the figure limits. Simulations were carried out on a desktop computer with an i7-14700 processor and 32GB of RAM running Ubuntu 22.04. Both Gurobi and the proposed MIQP solver were configured to use up to 16 threads. Results were averaged over 5 trials.

Results for the warm start case (labeled as ``Ours HZ WS") do not include the first time step of the simulation. This was done to show the effect of warm starting when compared to the ``Ours HZ" case, as the first time step of the former was not warm started.

In all cases, the proposed MIQP solver using hybrid zonotope constraints was able to converge faster on average and in the worst case than Gurobi using either hybrid zonotope or H-rep constraints. Gurobi was faster using a hybrid zonotope formulation than H-rep for some, but not all, examples. The proposed solver was always faster (up to two orders of magnitude) using hybrid zonotopes due to its structure-exploiting construction.

Table~\ref{tab:bnb-iter} shows the number of branch and bound iterations (i.e., the number of QP sub-problems that were solved) using the proposed MIQP solver for the data shown in Fig.~\ref{fig:sol-times}. For all maps, the solver converged in fewer iterations using hybrid zonotope representations of $\mathcal{F}$ than using halfspace representations. This shows the impact of the tighter hybrid zonotope convex relaxations when compared to relaxations using the Big-M method, as discussed in Sec.~\ref{sec:conv-relax}. 

Table~\ref{tab:qp-sol-time} shows the average QP sub-problem solution times for the data given in Fig.~\ref{fig:sol-times} using the proposed solver. In all cases, QP solution times were reduced using the hybrid zonotope representation of $\mathcal{F}$ when compared to H-rep, again reflecting the structure-exploiting construction of the solver. This improvement was particularly pronounced for OGMs, as the number of equality constraints in the hybrid zonotope representation does not increase with map complexity.

\begin{table}[tb]
    \caption{Comparisons of average and maximum branch-and-bound (b\&b) iterations between constraint representations and using binary variable warm starting.}
 \setlength{\tabcolsep}{4pt}
    \centering
    \begin{tabular}{c | c c c c c | c c c c c }
        \toprule
         &  \multicolumn{5}{c}{\textbf{avg b\&b iter}} & \multicolumn{5}{c}{\textbf{max b\&b iter}}  \\
         \midrule
         \textbf{map} & (a) & (b) & (c) & (d) & (e) & (a) & (b) & (c) & (d) & (e) \\
         \midrule
         HZ WS & 21 & 9 & 12 & 68 & 45 & 161 & 31 & 59 & 1241 & 181 \\
         HZ & 32 & 14 & 62 & 222 & 125 & 199 & 63 & 159 & 2283 & 390 \\
         H-Rep & 81 & 54 & 123 & N/A & 133 & 573 & 191 & 403 & N/A & 554 \\
         \bottomrule
    \end{tabular}
    \label{tab:bnb-iter}
\end{table}

\begin{table}[tb]
    \caption{Comparisons of average QP solution times between constraint representations and using binary variable warm starting.}
 \setlength{\tabcolsep}{4pt}
    \centering
    \begin{tabular}{c | c c c c c }
        \toprule
         &  \multicolumn{5}{c}{\textbf{avg QP time [ms]}} \\
         \midrule
         \textbf{map} & (a) & (b) & (c) & (d) & (e) \\
         \midrule
         HZ WS & 4 & 2 & 1 & 4 & 11 \\
         HZ & 4 & 3 & 3 & 5 & 8 \\
         H-Rep & 19 & 11 & 166 & 6324 & 354 \\
         \bottomrule
    \end{tabular}
    \label{tab:qp-sol-time}
\end{table}

Table~\ref{tab:sol-times-vs-N} shows solution times vs. the MPC prediction horizon $N$. 
Solution times generally increased with increasing $N$, though the rate of increase was often much steeper for H-rep constraint representations and for Gurobi than it was for the proposed MIQP solver using hybrid zonotope constraints. Warm starting the binary variables additionally reduced the sensitivity of the solution time to the prediction horizon.   

\begin{table}[tb]
 \setlength{\tabcolsep}{3pt}
    \centering
    \caption{Average and maximum MPC solution times for the double integrator example as a function of MPC horizon $N$.}
    \begin{tabular}{c|c|c c c c|c c c c}
    \toprule
        & & \multicolumn{4}{c}{\textbf{avg sol time [sec]}} & \multicolumn{4}{c}{\textbf{max sol time [sec]}} \\
        \midrule
        & \textbf{solver and} & \multicolumn{4}{c}{N} & \multicolumn{4}{c}{N} \\ 
        \textbf{map} & \textbf{constraint} & 5 & 10 & 15 & 20 & 5 & 10 & 15 & 20 \\
        \midrule
        \multirow{5}{*}{(a)} & Ours HZ WS & 0.00 & 0.01 & 0.02 & 0.06 & 0.01 & 0.04 & 0.12 & 0.61 \\
        & Ours HZ &  0.00 & 0.01 & 0.02 & 0.08 & 0.01 & 0.04 & 0.11 & 0.69 \\
        & Ours H-rep & 0.01 & 0.05 & 0.16 & 0.49 & 0.05 & 0.26 & 0.92 & 4.13 \\
        & Gurobi HZ & 0.03 & 0.12 & 0.28 & 0.68 & 0.06 & 0.25 & 0.52 & 4.14 \\
        & Gurobi H-rep & 0.02 & 0.08 & 0.21 & 0.42 & 0.04 & 0.14 & 0.39 & 0.80 \\
        \midrule
        \multirow{5}{*}{(b)} & Ours HZ WS & 0.00 & 0.01 & 0.01 & 0.02 & 0.00 & 0.02 & 0.04 & 0.06 \\
        & Ours HZ & 0.00 & 0.01 & 0.01 & 0.03 & 0.01 & 0.02 & 0.04 & 0.08 \\
        & Ours H-rep & 0.01 & 0.04 & 0.09 & 0.19 & 0.02 & 0.11 & 0.25 & 0.49 \\
        & Gurobi HZ & 0.01 & 0.05 & 0.12 & 0.31 & 0.03 & 0.10 & 0.33 & 0.98 \\
        & Gurobi H-rep &  0.02 & 0.07 & 0.18 & 0.40 & 0.03 & 0.19 & 0.46 & 1.62 \\
        \midrule
        \multirow{5}{*}{(c)} & Ours HZ WS & 0.00 & 0.00 & 0.01 & 0.01 & 0.01 & 0.02 & 0.03 & 0.04 \\
        & Ours HZ & 0.00 & 0.01 & 0.03 & 0.05 & 0.01 & 0.03 & 0.06 & 0.10 \\
        & Ours H-rep & 0.18 & 0.76 & 1.82 & 3.41 & 0.65 & 2.55 & 5.41 & 8.93 \\
        & Gurobi HZ &  0.06 & 0.26 & 0.70 & 1.32 & 0.12 & 0.63 & 1.49 & 3.29 \\
        & Gurobi H-rep & 0.26 & 1.09 & 3.23 & 6.81 & 0.32 & 1.78 & 6.85 & 21.15 \\
        \midrule
        \multirow{5}{*}{(d)} & Ours HZ WS & 0.01 & 0.07 & 0.08 & 0.04 & 0.05 & 0.41 & 0.78 & 0.07 \\
        & Ours HZ & 0.01 & 0.11 & 0.15 & 0.12 & 0.04 & 0.95 & 1.04 & 1.13 \\
        & Ours H-rep & 19.21 & N/A & N/A & N/A & N/A & N/A & N/A & N/A \\
        & Gurobi HZ & 0.41 & 1.61 & 3.72 & 6.04 & 0.94 & 4.33 & 11.97 & 33.21 \\
        & Gurobi H-rep & 5.35 & 21.40 & N/A & N/A & 7.30 & 40.63 & N/A & N/A \\
        \midrule
        \multirow{5}{*}{(e)} & Ours HZ WS & 0.03 & 0.06 & 0.07 & 0.10 & 0.04 & 0.17 & 0.15 & 0.38 \\
        & Ours HZ & 0.02 & 0.06 & 0.11 & 0.19 & 0.09 & 0.19 & 0.29 & 0.53 \\
        & Ours H-rep & 0.57 & 2.44 & 3.97 & 7.61 & 2.50 & 14.49 & 14.98 & 34.03 \\
        & Gurobi HZ & 0.08 & 0.33 & 0.79 & 1.35 & 0.12 & 0.62 & 2.20 & 4.41 \\
        & Gurobi H-rep & 0.40 & 1.65 & 4.16 & 23.80 & 0.49 & 3.01 & 7.50 & N/A \\
        \bottomrule
    \end{tabular}
    \label{tab:sol-times-vs-N}
\end{table}

The effects of the reachability calculations and exploiting the box constraint structure of the hybrid zonotope constraints are quantified in Table~\ref{tab:reach-diag}. The baseline case corresponds to the proposed MIQP solver without warm starting for an MPC prediction horizon of $N=15$ steps. The ``no reach" case is constructed by setting $d_{max} = \infty$ (see Definition~\ref{def:reachability}). This makes all regions reachable from each other and from the initial condition, effectively eliminating the effects of the reachability calculations. The ``no diag" case was constructed by forcing the QP solver to compute Cholesky factorizations of the $\Phi_k$ matrices in \eqref{eq:Phi_k}, even though these matrices are diagonal. Both the reachability logic and the box constraint structure exploitation contribute significantly to the overall performance of the MIQP solver.

\begin{table}[tb]
    \caption{Effect of reachability constraints and diagonal matrix exploitation on the overall solution time for an MPC prediction horizon of $N=15$. The ``baseline" case corresponds to the proposed MIQP solver using a hybrid zonotope representation of the obstacle-free space (i.e., ``Ours HZ").}
 \setlength{\tabcolsep}{3pt}
    \centering
    \begin{tabular}{c | c c c c c | c c c c c }
        \toprule
         &  \multicolumn{5}{c}{\textbf{avg sol time [sec]}} & \multicolumn{5}{c}{\textbf{max sol time [sec]}}  \\
         \midrule
         \textbf{map} & (a) & (b) & (c) & (d) & (e) & (a) & (b) & (c) & (d) & (e) \\
         \midrule
         baseline & 0.02 & 0.01 & 0.03 & 0.15 & 0.11 & 0.11 & 0.04 & 0.06 & 1.04 & 0.29 \\
         no reach & 0.08 & 0.02 & 0.26 & 1.69 & 0.19 & 0.63 & 0.10 & 0.70 & 10.09 & 0.48 \\
         no diag & 0.07 & 0.03 & 0.06 & 1.25 & 0.19 & 0.36 & 0.10 & 0.13 & 11.91 & 0.51 \\
         no reach & \multirow{2}{*}{0.25} & \multirow{2}{*}{0.06} & \multirow{2}{*}{0.67} & \multirow{2}{*}{N/A} & \multirow{2}{*}{0.47} & \multirow{2}{*}{1.90} & \multirow{2}{*}{0.26} & \multirow{2}{*}{1.83} & \multirow{2}{*}{N/A} & \multirow{2}{*}{1.24} \\
         or diag & & & & & & & & & & \\
         \bottomrule
    \end{tabular}
    \label{tab:reach-diag}
\end{table}

\subsection{Processor-in-the-Loop Testing} \label{sec:results-pil}
Processor-in-the-loop (PIL) testing was conducted to assess the feasibility of applying the proposed MIQP solver to motion planning problems in a real-time context on embedded hardware. With the exception of the differences noted below, the same maps, models, and parameterizations given in Sec.~\ref{sec:results-sim} are used here. 

The PIL tests were executed using a NVIDIA Jetson AGX Orin embedded computer which has a 2.2~GHz, 12 core CPU. A ROS2 implementation was used for the MPC controller using the proposed MIQP solver. To facilitate real-time implementation, the optimal control input at time-step $k=1$, $\mathbf{u}_1$, was stored after the solver returns and applied at the start of next MPC iteration. The control input at $k=0$ was constrained accordingly as described in Sec.~\ref{sec:problem-formulation}. The MPC horizon was set to $N=15$. The MIQP solver was configured to use up to 6 threads. Warm starting was used for all time steps of the PIL test except for the first step. The MPC loop rate was set to 5~Hz and the dynamics were updated at 100~Hz to approximate continuous time. For parity with the examples from Sec.~\ref{sec:results-sim}, the simulated dynamics were executed $5$ times faster than real-time such that $1$~sec of simulated time passed for each loop of the controller.

The simulated trajectories from the PIL tests are compared to the optimal simulated trajectories from Sec.~\ref{sec:results-sim} in Fig.~\ref{fig:traj-PIL}. Observed discrepancies between the trajectories result from the solver returning before convergence in the PIL tests to meet the required 5~Hz update rate. Discrepancies can also occur when there are multiple feasible motion plans with costs within $\epsilon_a$ or $\epsilon_r$ of each other. In map \textbf{(a)}, the solver returned before convergence in 8 out of 49 iterations with a worst-case guaranteed optimality gap (i.e., $|j_+-j_-|/|j_+|$) of $77.6\%$. In maps \textbf{(b)} and \textbf{(c)}, the solver converged on all iterations. In map \textbf{(d)}, the solver returned a sub-optimal solution on 11 of 39 iterations with a worst-case guaranteed optimality gap of $28.4\%$. In map \textbf{(e)}, 9 of 39 solutions were sub-optimal with a worst-case guaranteed optimality gap of $99.7\%$. 

A metric for the quality of the trajectories in Fig.~\ref{fig:traj-PIL} is the integrated stage costs $J_{st} = \sum_n (\mathbf{x}_n-\mathbf{x}_r)^T Q_k (\mathbf{x}_n-\mathbf{x}_r) + \mathbf{u}_n^T R_k \mathbf{u}_n + q^r(\mathbf{y}_n)$ where $n$ is the simulation time step. These costs are given for the optimal simulated (i.e., converged) and PIL trajectories in Table~\ref{tab:pil_stage_costs}. In all cases, the integrated stage costs for the PIL trajectories are comparable to those of the optimal simulations. 

\begin{table}[tb]
    \caption{Integrated stage costs for the optimal simulated and PIL trajectories.}
    \centering
    \begin{tabular}{c|c|c|c|c|c}
    \toprule
         \textbf{map} & (a) & (b) & (c) & (d) & (e) \\ \midrule
         sim & 2320.5 & 2320.6 & 14676.8 & 635.5 & 1635.2 \\
         PIL & 2326.2 & 2321.1 & 14679.0 & 652.7 & 1638.1 \\ \bottomrule
    \end{tabular}
    \label{tab:pil_stage_costs}
\end{table}

\begin{figure*}[t!]
\centering
\begin{minipage}{0.6\textwidth}
\includegraphics[width=\textwidth]{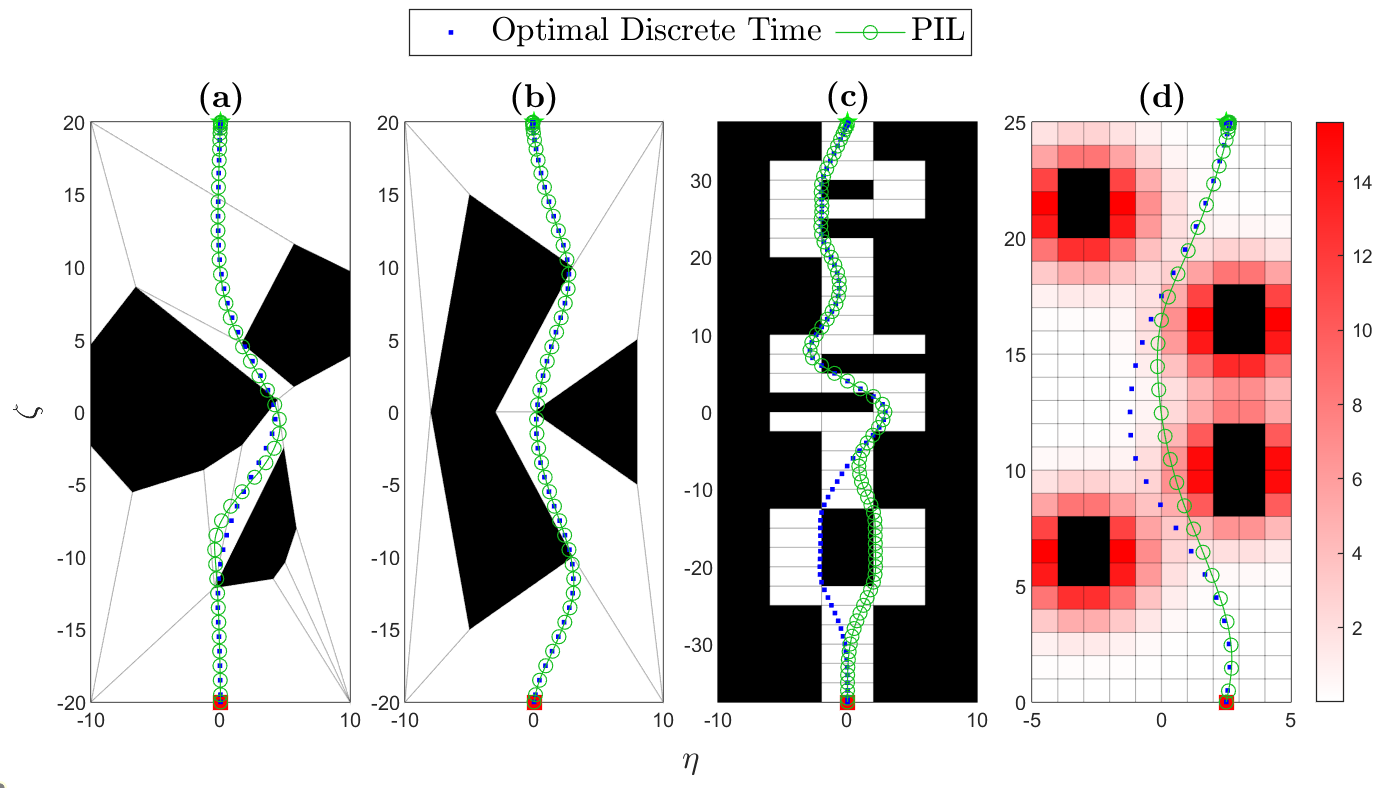}
\end{minipage}
\begin{minipage}{0.3\textwidth}
\includegraphics[width=\textwidth]{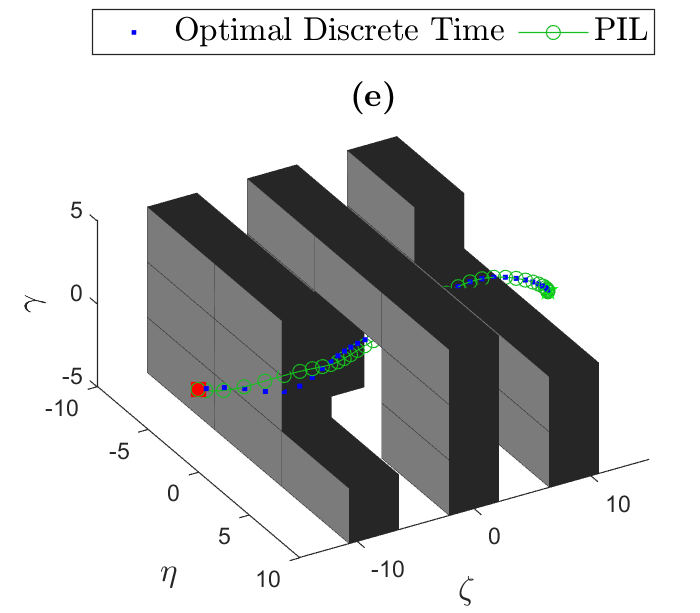}
\end{minipage}
\caption{Processor-in-the-loop results for double integrator model compared to optimal trajectories from Fig.~\ref{fig:trajectories}.}
\label{fig:traj-PIL}
\end{figure*}
\section{Conclusions}
We presented an MPC formulation and a corresponding MIQP solver applicable to AV motion planning problems. General polytopic partitions of the obstacle-free space and occupancy grid maps were both considered. Risk-aware planning was performed by incorporating occupancy probabilities from an occupancy grid map into the MPC cost function. A hybrid zonotope representation of the obstacle-free space was used, and its structure was exploited within the MIQP solver. Significant improvements in optimization time were observed using our solver when compared to a state-of-the-art commercial solver and a halfspace representation of the obstacle-free space. The solver was also shown to perform well when running in a real-time context on embedded hardware.

Future work will extend these techniques to dynamic environments where obstacle positions and occupancy probabilities vary with time step. Performance will be evaluated in dynamic environments using high-fidelity simulations and/or robotics experiments.

\bibliography{bibitems}
\bibliographystyle{ieeetr}

\end{document}